%% file: main.tex
\documentclass{article}

\usepackage{arxiv}

\usepackage[utf8]{inputenc} 
\usepackage[T1]{fontenc}    
\usepackage{hyperref}       
\usepackage{url}            
\usepackage{booktabs}       
\usepackage{nicefrac}       
\usepackage{microtype}      
\usepackage{lipsum}
\usepackage{fancyhdr}       
\usepackage{doi}

\usepackage{amsmath}
\usepackage{amsfonts}
\usepackage{amssymb}
\usepackage{amsthm} 
\usepackage{graphicx}
\usepackage{dcolumn}
\usepackage{bm}

\usepackage{tikz} 
\usetikzlibrary{scopes}
\usepackage{pgfplots}
\usetikzlibrary{arrows.meta}

\usetikzlibrary{external} 
\pgfplotsset{compat=1.14}

\definecolor{grays1}{RGB}{70, 70, 70}

\definecolor{blues1}{RGB}{198, 219, 239}
\definecolor{blues2}{RGB}{158, 202, 225}
\definecolor{blues3}{RGB}{107, 174, 214}
\definecolor{blues4}{RGB}{49, 130, 189}
\definecolor{blues5}{RGB}{8, 81, 156}

\definecolor{reds1}{RGB}{239,198, 198}
\definecolor{reds2}{RGB}{225,158, 158}
\definecolor{reds3}{RGB}{214,107, 107}
\definecolor{reds4}{RGB}{189,70,70}
\definecolor{reds5}{RGB}{156, 8, 81}

\definecolor{greens1}{RGB}{198, 239,198}
\definecolor{greens2}{RGB}{158, 225, 158}
\definecolor{greens3}{RGB}{107, 214, 130}
\definecolor{greens4}{RGB}{50, 189, 70}
\definecolor{greens5}{RGB}{8, 156, 51}

\definecolor{trans1}{RGB}{49, 130, 189}
\definecolor{trans2}{RGB}{77,118,165}
\definecolor{trans3}{RGB}{105, 106,141}
\definecolor{trans4}{RGB}{133, 94, 117}
\definecolor{trans5}{RGB}{161, 82, 93}
\definecolor{trans6}{RGB}{189,70,69}

\newcommand{\sourcetime}{\ensuremath{t_{s}}}
\newcommand{\sourceposition}{\ensuremath{x_{s}}}
\newcommand{\peakfreq}{\ensuremath{\omega_{p}}}
\newcommand{\psiricker}{\ensuremath{\psi^{r}}}
\newcommand{\phiricker}{\ensuremath{\varphi^{r}}}
\newcommand{\psirickeromega}{\ensuremath{\psi^{r}_{\omega_{p}}}}
\newcommand{\phirickeromega}{\ensuremath{\varphi^{r}_{\omega_{p}}}}
\newcommand{\psirickeromegadef}[1]{\ensuremath{\psi^{r}_{\omega_{p}=#1}}}
\newcommand{\phirickeromegadef}[1]{\ensuremath{\varphi^{r}_{\omega_{p}=#1}}}
\newcommand{\bandwidth}{\ensuremath{\Delta^{\varphi^r}_{n,\peakfreq}}}

\newcommand{\bandwidthnomega}[2]{\ensuremath{\Delta^{\varphi^r}_{#1,#2}}}
\newcommand{\stdgauss}{\ensuremath{\sigma_{\peakfreq}}}
\newcommand{\stdgaussomega}[1]{\ensuremath{\sigma_{#1}}}

\title{Transferring the inhomogeneous wave equation into a homogeneous equation}

\author{ \href{https://orcid.org/0000-0002-4411-3635}{Marcos V. C. Henriques} \\
	Departamento de Ci\^encias Exatas e Tecnologia da Informa\c{c}\~ao\\
	Universidade Federal Rural do Semi-\'Arido\\
	Angicos, Brazil \\
	\texttt{viniciuscandido@ufersa.edu.br} \\
}

\begin{document}

\maketitle
    
\begin{abstract}
The inhomogeneous wave equation, triggered by point sources, forms the basis for the most modern computational techniques of seismic inversion. In this work, we propose to transfer the inhomogeneous wave equation into a homogeneous equation. We show that one can suppress the wavelet-based source term from the inhomogeneous equation in favour of setting the initial time derivative condition of the wavefield as a scaled wavelet of the same type. With the homogeneous wave equation, one can slightly reduce the computational cost of numerical modeling.
\end{abstract}


\setcounter{subsection}{9}

\section{\label{sec:level1}Introduction}

Many physical phenomena involve the propagation of wavefronts generated by well-localized sources of very short duration. Examples can be found in acoustics \cite{habets2006room,allen1979image,ward2001reproduction}, electromagnetism \cite{tiwana2017point,vlaar1966field} and seismology \cite{vidale1988elastic,frankel1993three}. Modeling such sources is of particular interest in seismic exploration, where the seismic energy is provided in a controlled manner by a high-power short pulse on the surface, triggered by specialized devices such as thumpers or air guns \cite{evans1997handbook,meunier2011seismic}. In this case, the wavefronts are responsible for ``illuminating'' the geological layers interfaces from the subsurface. 

The simulation of wave propagation in continuous media, using a numerical solution of the wave equation, is the basis of modern techniques for solving seismic inversion problems. The inversion is a processing stage aimed at directly describing the properties of the rocks. Advanced inversion techniques such as Reverse Time Migration (RTM) \cite{baysal1983reverse,mcmechan1989review} and Full Waveform Inversion (FWI) \cite{virieux2009overview} massively make use of numerical solutions of the wave equation. The extremely high computational cost of these procedures causes, in most cases, the disregard of elastic effects.

In seismic modeling, an impulsive point source (IPS) is usually represented by a combination of a compact time function, such as a wavelet, and a spatial impulse function, such as the Dirac delta function \cite{cohen1979velocity,alford1974accuracy}. The generated wavefronts  are characteristic of the chosen wavelet. The wave energy is quickly added to the system as the time wavelet is in action. 

Mathematically speaking, setting up sources is not the only way to generate wavefronts. As introduced by d'Alembert in his notorious analytical solution of the homogeneous wave equation \cite{drecherches}, initial conditions of the wavefield can trigger the wave propagation. If there are no sources, and therefore no external forces, initial conditions make the system to have, at the initial time, the amount of energy that is going to be propagated. So, depending on the nature of the problem, to model the wave propagation, one could suppress the IPS term in favor of using only the initial conditions. That would be useful, for example, to slightly reduce the computational cost of a simulation.

The intent of this work is to show that, for the one-dimensional case, the waveforms generated by IPS's modeled with time wavelets can be reproduced by setting the initial time derivative of the wavefield as a scaled wavelet of the same type (section \ref{sec:wave-equation}). This 1D approach can be used as a base for the same methodology in realistic 2D and 3D cases.

\section{The wave equation\label{sec:wave-equation}}

The inhomogeneous wave equation in one dimension can be written as
\begin{equation}
c^{2}\partial_{x}^{2}u\left(x,t\right)-\partial_{t}^{2}u\left(x,t\right)=s\left(x,t\right),\;\;\;\;\;x\in\mathbb{R},\,t\in\mathbb{R}_{\geqslant0}\label{eq:waveequation}
\end{equation}
where the time-dependent scalar function $u=u\left(x,t\right)$
represents a physical quantity, depending on the problem in which the equation is applied, that produces waves that propagate with velocity $c$. For example, it can represent the transverse displacement of a string,
the electric or the magnetic field, or seismic waves that travel through
the Earth. Accordingly, the source term $s$, whose presence characterizes
the equation \eqref{eq:waveequation} as an inhomogeneous wave equation,
can represent an external force applied to a stretched string, point
sources of electromagnetic waves such as time-varying charge densities, 
or a seismic source. 

Without the source $s$, we have the homogeneous wave equation:
\begin{equation}
c^{2}\partial_{x}^{2}u\left(x,t\right)-\partial_{t}^{2}u\left(x,t\right)=0,\;\;\;\;\;x\in\mathbb{R},\,t\in\mathbb{R}_{\geqslant0}\label{eq:homowaveequation}
\end{equation}

We will refer to $u$ as the wavefield. The system of measurement
is irrelevant to the discussion of this work, but we choose the International System of Units (SI) for all the fundamental physical quantities.
However, the unit of the wavefield $u$ depends on the nature of the
problem. So, throughout the text, we will omit the unit when referring
to a value representing the wavefield in a specific position and time.

Let us denote the initial conditions by:
\begin{equation}
\begin{cases}
\begin{array}{lll}
u\left(x,0\right) & =f\left(x\right), & \;\;\;\;\;x\in\mathbb{R}\\
\partial_{t}u\left(x,0\right) & =g\left(x\right), & \;\;\;\;\;x\in\mathbb{R}
\end{array}\end{cases}\label{eq:initialconditions}
\end{equation}
where $f$ and $g$ are smooth functions. The expressions \eqref{eq:waveequation}
and \eqref{eq:initialconditions} make up the Cauchy problem for the
inhomogeneous wave equation. We will not worry about the boundary
conditions, since they do not matter for the main conclusions
of this work if the problem domain is infinite, that is, the medium is large enough.

The equation \eqref{eq:waveequation} is an example of a hyperbolic
partial diferential equation, and its solution under conditions \eqref{eq:initialconditions}
is given by \cite{miersemann2012partial}:
\begin{eqnarray}
u\left(x,t\right) & = & \cfrac{1}{2}\left[f\left(x+ct\right)+f\left(x-ct\right)\right]+\cfrac{1}{2c}\int_{x-ct}^{x+ct}g\left(x'\right)dx'\nonumber \\
 &  & +\,\cfrac{1}{2c}\int_{0}^{t}\int_{x-c\left(t-t'\right)}^{x+c\left(t-t'\right)}s\left(x',t'\right)dx'dt'\label{eq:solution_wave_equation}
\end{eqnarray}
where the first two terms of the right side make up the d'Alembert
solution for the homogeneous wave equation \eqref{eq:homowaveequation},
while the third term accounts for the source effects on the wavefield.
They are a direct consequence of the \textit{Principle of Causality}
\cite{drabek2014elements}, which ensures that the value of the solution
at a point$\left(x,t\right)$ is only influenced by the values that
are within its dependency domain, that is, its \textit{past light
cone}. In other words, the effects of a cause cannot influence a point
that they did not have time to reach, considering the velocity of
propagation $c$. So, as one can see in equation \eqref{eq:solution_wave_equation},
if we do not consider sources (third term), the value of $u\left(x,t_{0}+t\right)$,
for any $x$, $t_{0}$ and $t$, depends only on the value of $u\left(x-ct,t_{0}\right)$
and $u\left(x+ct,t_{0}\right)$ or on the values of $\left.\partial_{t}u\right|_{t=t_{0}}$
inside the interval $x-ct$ to $x+ct$. 

Let us refer to the three terms of the solution \eqref{eq:solution_wave_equation}
as $F\left(x,t\right)$, $G\left(x,t\right)$ and $S\left(x,t\right)$,
respectively:
\begin{eqnarray}
F(x,t) & = & \cfrac{1}{2}\left[f\left(x+ct\right)+f\left(x-ct\right)\right]\mbox{,}\label{eq:Fdefinition}\\
G(x,t) & = & \cfrac{1}{2c}\int_{x-ct}^{x+ct}g\left(x'\right)dx'\mbox{,}\label{eq:Gdefinition}\\
S\left(x,t\right) & = & \cfrac{1}{2c}\int_{0}^{t}\int_{x-c\left(t-t'\right)}^{x+c\left(t-t'\right)}s\left(x',t'\right)dx'dt'\label{eq:Sdefinition}
\end{eqnarray}

so that:
\[
u\left(x,t\right)=F(x,t)+G(x,t)+S(x,t)
\]

\subsection{Impulsive Point source}

Especially in seismic exploration, sources are often modeled as impulsive
point sources (IPS). Although its time profile is not known in a
real survey situation, an IPS $s(x,t)$ is usually designed by using
an integrable and continuous wavelet function $\psi$  in the following
way \cite{cohen1979velocity,alford1974accuracy}:
\begin{equation}
s_{\sourceposition,\sourcetime}(x,t)=\psi(t-\sourcetime)\delta(x-\sourceposition)\label{eq:single-pulse-source}
\end{equation}
where $\delta$ is the Dirac delta function, $\sourcetime$ shifts
the wavelet in time and can be interpreted as the instant in which
there is the maximum rate of energy release, and $\sourceposition$
is the location of the source. The use of the Dirac delta function
as the spatial part of the model explicits that we are dealing with
a point source.

The use of wavelet functions is justified by the fact that they are
well localized in both time and frequency, can be easily scaled and
translated, and have a zero mean ($\intop_{-\infty}^{\infty}\psi\left(t\right)dt=0$)
\cite{mallat1999wavelet}. Not all wavelets have compact support,
but they usually are rapidly decreasing functions and, therefore,
vanish at infinity. It is desirable too to use square-integrable functions, that is, satisfying to $\intop_{-\infty}^{\infty}\left|\psi\left(t\right)\right|^{2}dt<\infty$.

Let us now insert the model \eqref{eq:single-pulse-source} into the ``source solution'' \eqref{eq:Sdefinition} of the wave equation: 
\begin{equation}
S_{c,\sourceposition,\sourcetime}^{\psi}\left(x,t\right)=\cfrac{1}{2c}\int_{0}^{t}\int_{x-c\left(t-t'\right)}^{x+c\left(t-t'\right)}\psi\left(t'-\sourcetime\right)\delta\left(x'-\sourceposition\right)dx'dt'\label{eq:Sinit}
\end{equation}  

It can be shown (appendix \ref{sec:derivS}), that this integral leads to:

\begin{equation}
S_{c,\sourceposition,\sourcetime}^{\psi}\left(x,t\right)=\cfrac{1}{2c}\left[\varphi\left(t-\sourcetime-\cfrac{\left|x-\sourceposition\right|}{c}\right)-\varphi\left(-\sourcetime\right)\right]\label{eq:Ssolution}
\end{equation}
where $\varphi\left(x\right)$ is the antiderivative of the wavelet
$\psi\left(x\right)$ on every closed interval (in practice, the ``indefinite
integral'' of $\psi\left(x\right)$), also having a zero mean. We
can interpret $\varphi$ as the waveform generated by the source.
As it will be shown further, the term $\varphi\left(t-\sourcetime-\left|x-\sourceposition\right|/c\right)$
takes the form of two wavefronts traveling in the opposite direction
to each other.

The solution \eqref{eq:Ssolution} still reveals that $\sourcetime$
has to have a positive minimum value that ensures that it is out of
the compact subset in which $\varphi\left(-t\right)$ is significative;
otherwise, the term $\varphi\left(-\sourcetime\right)$ would add
a significant constant to the $S^{\psi}$ function, causing it not to have zero mean. This guarantees that the brief supplying of energy
to the system by the source starts at a time $t>t_{0}$, since the
relevant source activity starts before $\sourcetime$.

\subsection{Initial time derivative condition}

Consider that the initial time derivative condition $g\left(x\right)$
is given by a wavelet $\psi(x)$, scaled and translated as:
\begin{equation}
g_{c,\sourceposition}(x)=\cfrac{1}{c}\,\psi\left(\cfrac{x-\sourceposition}{c}\right)\label{eq:gwavelet}
\end{equation}
Note that this wavelet is dependent on position $x$ instead of time
$t$ as in the definition \eqref{eq:single-pulse-source} of the source
$s$. This function is scaled in a different way with which the wavelet
is conventionally scaled $\psi_{c,x_{0}}(x)=\frac{1}{\sqrt{c}}\psi\left(\frac{x-x_{0}}{c}\right)$
\cite{mallat1999wavelet}. The reason for this choice will become clear later.

With the definition \eqref{eq:gwavelet}, the solution \eqref{eq:Gdefinition}
becomes:
\[
G_{c,\sourceposition}^{\psi}\left(x,t\right)=\cfrac{1}{2c}\int_{x-ct}^{x+ct}\cfrac{1}{c}\,\psi\left(\cfrac{x'-\sourceposition}{c}\right)dx'
\]

It is easy, by substitution of variables, to evaluate this integral
and obtain (appendix \ref{sec:derivG})
\begin{equation}
G_{c,\sourceposition}^{\psi}\left(x,t\right)=\cfrac{1}{2c}\left[\varphi\left(\cfrac{x-\sourceposition}{c}+t\right)-\varphi\left(\cfrac{x-\sourceposition}{c}-t\right)\right]\label{eq:Gsolution}
\end{equation}
in which $\varphi\left(x\right)$ is the antiderivative of the wavelet $\psi\left(x\right)$.
One can note that this solution, as the $S^{\psi}$, produces two
wavefronts moving in opposite directions to each other. Additionally,
the spatial part of this function has mirror symmetry with respect
to $\sourceposition$ position. That is the difference with the $F$
solution, as defined by \eqref{eq:Fdefinition}, which generates waveforms
that are simply reduced copies of the initial condition $f$: if one
sets $f$ as an asymmetric function, the $F$ solution will also be
asymmetric. This discourages us from using $f$ to emulate an IPS, since
the opposing waveforms produced by a impulsive source have reflection
symmetry to each other.

\subsection{Comparison of waveforms}

We are going now to show that the expressions \eqref{eq:Ssolution}
and \eqref{eq:Gsolution} for $S$ and $G$, respectively, in the
special case in which $\varphi$ is an odd function and vanishes at
infinity, generate almost equal waveforms, although out of phase with
each other.
\newtheorem{prop}{Proposition}
\begin{prop}
Let $\varphi$ be an odd function that vanishes at infinity. Let $G'\left(\xi,t\right)$
and $S'\left(\xi,t\right)$ be functions defined by \label{prop:propI}
\begin{equation}
G'\left(\xi,t\right)=\alpha\bigl[\varphi\left(\xi+t\right)-\varphi\left(\xi-t\right)\bigr]\label{eq:Gproposition}
\end{equation}
\begin{equation}
S'\left(\xi,t\right)=\alpha\bigl[\varphi\left(t-t_{0}-\left|\xi\right|\right)-\varphi\left(-t_{0}\right)\bigr]\label{eq:Sproposition}
\end{equation}
with $\alpha\in\mathbb{R}_{>0}$, $\xi\in\mathbb{R}$, $t\in\mathbb{R}_{\geqslant0}$
and $t_{0}\in\mathbb{R}_{\geqslant0}$. So, for all $\epsilon\in\mathbb{R}_{>0}$,
there exists a pair $(\tau,\tau_{0})\in\mathbb{R}_{\geqslant0}^{2}$
for which\textcolor{blue}{{} }
\[
\bigl|G'\left(\xi,t\right)-S'\left(\xi,t+t_{0}\right)\bigr|<\epsilon
\]
for all $\xi\in\mathbb{R}$, $t>\tau$ and $t_{0}>\tau_{0}$.
\end{prop}
\begin{proof}
From the $G'$ and $S'$ definitions:
\begin{equation}
\bigl|G'\left(\xi,t\right)-S'\left(\xi,t+t_{0}\right)\bigr|=\alpha\bigl|\varphi\left(\xi+t\right)-\varphi\left(\xi-t\right)-\varphi\left(t-\left|\xi\right|\right)+\varphi\left(-t_{0}\right)\bigr|\label{eq:difGSI}
\end{equation}

Let us first consider the subset $\xi\geq0$, in which $\varphi\left(t-\left|\xi\right|\right)=\varphi\left(t-\xi\right)$.
As $\varphi$ is an odd function, $\varphi\left(t-\xi\right)=-\varphi\left(\xi-t\right)$,
and two terms cancel each other on the equation \eqref{eq:difGSI}:
\begin{equation}
\bigl|G'\left(\xi,t\right)-S'\left(\xi,t+t_{0}\right)\bigr|=\alpha\bigl|\varphi\left(\xi+t\right)+\varphi\left(-t_{0}\right)\bigr|\label{eq:gsproofI}
\end{equation}
Since $\varphi$ is a function that vanishes at infinity, given any
$\epsilon>0$ and $\alpha>0$, one can choose $\tau$ e $\tau_{0}$
such that
\begin{equation}
\bigl|\varphi\left(\xi+\tau\right)\bigr|<\cfrac{\epsilon}{2\alpha}\;,\;\;\;\;\;\;\;\bigl|\varphi\left(-\tau_{0}\right)\bigr|<\cfrac{\epsilon}{2\alpha}\label{eq:gsproofII}
\end{equation}

to any $\xi\geq0$. Therefore, applying to the equation \eqref{eq:gsproofI}
the relations \eqref{eq:gsproofII} and the subadditivity property
of absolute value:
\begin{equation}
\bigl|G'\left(\xi,t\right)-S'\left(\xi,t+t_{0}\right)\bigr|<\epsilon,\;\;\;\;\;\;\;\forall\xi\in\mathbb{R}_{\geqslant0}\label{eq:gsproofIII}
\end{equation}
for all $t>\tau$ and $t_{0}>\tau_{0}$.

For the subset $\xi<0$, in which $\varphi\left(t-\left|\xi\right|\right)=\varphi\left(t+\xi\right)$,
we get:
\[
\bigl|G'\left(\xi,t\right)-S'\left(\xi,t+t_{0}\right)\bigr|=\alpha\bigl|-\varphi\left(\xi-t\right)+\varphi\left(-t_{0}\right)\bigr|
\]

Using a procedure similar to that of the case $\xi\geq0$, we achieve
the same upper limit as \eqref{eq:gsproofIII}:
\begin{equation}
\bigl|G'\left(\xi,t\right)-S'\left(\xi,t+t_{0}\right)\bigr|<\epsilon,\;\;\;\;\;\;\;\forall\xi\in\mathbb{R}_{<0}\label{eq:gsproofIV}
\end{equation}

Therefore, the proposition is valid in the whole set of real numbers:
\[
\bigl|G'\left(\xi,t\right)-S'\left(\xi,t+t_{0}\right)\bigr|<\epsilon,\;\;\;\;\;\;\;\forall\xi\in\mathbb{R}
\]
\end{proof}

Realize that $G\left(x,t\right)$ and $S\left(x,t+\sourcetime\right)$,
according to the solutions \eqref{eq:Gsolution} and \eqref{eq:Ssolution},
can be represented in the forms \eqref{eq:Gproposition} and \eqref{eq:Sproposition},
respectively, with $\alpha=\left(2c\right)^{-1}$, $\xi=\left(x-\sourceposition\right)/c$,
$c\geq0$ and $\sourcetime=t_{0}$. Note also that this proposition
requires that $G$ and $S$ have the same factor $\alpha$, what explains
our definition \eqref{eq:gwavelet}.

The proposition \ref{prop:propI} reveals that, when solving the wave
equation \eqref{eq:waveequation}, if one chooses \eqref{eq:gwavelet}
as the initial condition $g\left(x\right)$ for the time derivative
of the wavefield, and sets $f\left(x\right)=0$ and $s(x,t)=0$, as
$t\rightarrow\infty$, one gets a solution with waveforms almost identical
to those that would be formed if, instead, the function $s(x,t)$
was defined as an impulsive point source (IPS) modeled as in equation \eqref{eq:single-pulse-source}.
This is the main achievement of this work. It is desirable that $\varphi(t)$
be a rapidly decreasing function, so that the convergence be fast.

\section{Application of the Ricker wavelet\label{sec:ricker-wavelet}}

A very commonly used continuous wavelet for modeling short pulses
is the Ricker wavelet, which proved to be very suitable for modeling
seismic sources \cite{gholamy2014ricker}. A common definition in time domain is the negative second derivative of a gaussian function\cite{wang2015generalized, wang2015frequencies, wang2015ricker}:
\begin{equation}
\psirickeromega\left(t\right)=\left(1-\cfrac{\omega_{p}^{2}t^{2}}{2}\right)e^{-\omega_{p}^{2}t^{2}/4}\label{eq:rickerdef}
\end{equation}
where $t$ is the time in seconds and $\omega_{p}$ is the peak frequency,
that is, the most energetic frequency, in radians per second. Note
that, like other wavelets, this function has zero mean $(\intop_{-\infty}^{\infty}\psirickeromega\left(t\right)dt=0)$.
Its antiderivative function $\varphi$ corresponds to the first derivative
of a gaussian:
\begin{equation}
\phirickeromega\left(t\right)=\,t\,e^{-\omega_{p}^{2}t^{2}/4}\label{eq:rickerphidef}
\end{equation}

\begin{figure}
\input{figures/rickerwavelet.tex}

\caption{(a) The Ricker wavelet $\psiricker$, as defined by \eqref{eq:rickerdef},
with $\omega_{p}=1\,\mbox{rad/s}$, and its antiderivative $\phiricker$
as defined by \eqref{eq:rickerphidef}. (b) Controlling the size of
the compact subset in which $\varphi^{r}$ is significative.\label{fig:ricker-wavelet}}
\end{figure}
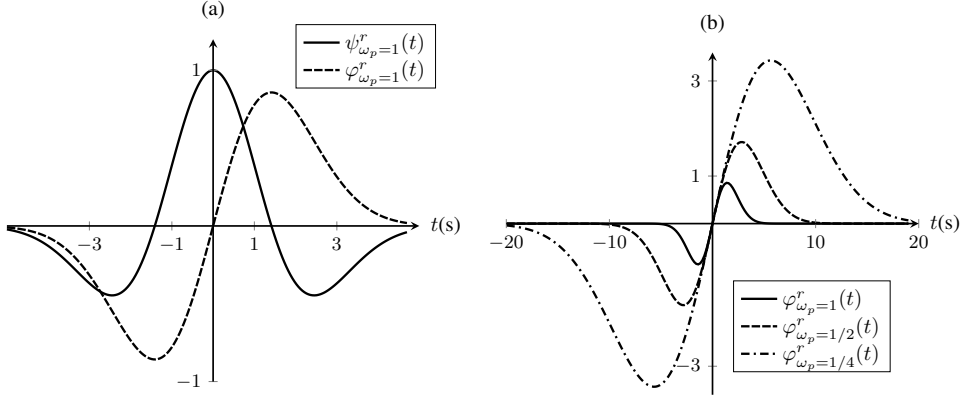

As one can see in figure \ref{fig:ricker-wavelet}-a, $\psiricker$
is an even function while $\phiricker$ is an odd function, that is,
$\psiricker$ is symmetric while $\phiricker$ is antisymmetric for
any $\omega_{p}$. Even though $\psiricker$ and $\phiricker$ are
not compactly supported, they are smooth functions, vanishing at infinity
and rapidly decreasing $\left(\mbox{o}\left(\left|t\right|^{-N}\right),\,\forall N\in\mathbb{R}_{>0}\right)$.
Therefore, $\phiricker$ meets the requirements of proposition \ref{prop:propI}.
By changing the peak frequency $\peakfreq$, one can manipulate the
size of the compact subset in which $\phiricker$ and its first derivative
$\psiricker$ are \textquotedblleft non-negligible\textquotedblright, as it is shown in figure \ref{fig:ricker-wavelet}-b.
Let us define the extension of the chosen significative subset as
\begin{equation}
\bandwidth=n\,\stdgauss\label{eq:bandwidth}
\end{equation}
in which $\stdgauss$ is the standard deviation of the corresponding
Gaussian function whose first derivative is $\phirickeromega$, and
$n$ is an integer positive number that sets how many standard deviations
are considered. In appendix \ref{sec:CalculusStdGauss} it is shown that $\stdgauss=\sqrt{2}/\peakfreq$.
Since the ratio between the value of the zero-centered gaussian function
at $t=4\,\stdgauss$ and its maximum value is of the order of $10^{-4}$,
we can say that good choices for $n$ are the ones starting at 4.

The reference frame for the source position is irrelevant, which alllow
us to set $\sourceposition=0$ and omit this parameter from now on.
So, the solutions $S^{\psi}$ and $G^{\psi}$, as stated by \eqref{eq:Ssolution}
and \eqref{eq:Gsolution}, employing the Ricker wavelet, can be written
as

\begin{equation}
S_{c,\peakfreq,\sourcetime}^{r}\left(x,t\right)=\cfrac{1}{2c}\left[\left(-\frac{\left|x\right|}{c}+t-\sourcetime\right)e^{-\omega_{p}^{2}\left(\left|x\right|/c-t+\sourcetime\right)^{2}/4}+\sourcetime e^{-\sourcetime^{2}\peakfreq^{2}/4}\right]\label{eq:Sricker}
\end{equation}
\begin{equation}
G_{c,\peakfreq}^{r}\left(x,t\right)=\cfrac{1}{2c^{2}}\left[(ct-x)e^{-\peakfreq^{2}(ct-x)^{2}/4c^{2}}+(ct+x)e^{-\omega_{p}^{2}(ct+x)^{2}/4c^{2}}\right]\label{eq:Gricker}
\end{equation}

Let us set, for convenience, $\peakfreq=1\mbox{ rad/s}$ and $c=1\mbox{\,m/s}$.
With this choice for $\peakfreq$, a suitable value for $\sourcetime$
would be $\bandwidthnomega{4}{1}/2=4\,\stdgaussomega{1}=4\sqrt{2}\,\mbox{s}\simeq5.66\,\mbox{s}$,
since we have to guarantee that half of the significant range of the
time wavelet fits between $t=0$ and $t=t_{s}$. In figures \ref{fig:evolutions}-a
and \ref{fig:evolutions}-b are shown the evolution of \eqref{eq:Sricker}
and \eqref{eq:Gricker}, respectively, with uniform time intervals.
It is remarkable how the two solutions, that start differing from
each other, converge quickly to form waveforms of the same shape,
amplitude and wavelength. To quantitatively measure this, let us define
the following difference operator:
\begin{equation}
d_{c,\peakfreq,\sourcetime}^{\,G^{r},S^{r}}\left(x,t\right)=\left|G_{c,\peakfreq}^{r}\left(x,t\right)-S_{c,\peakfreq,\sourcetime}^{r}\left(x,t+\sourcetime\right)\right|\label{eq:differenceOperator}
\end{equation}
where $S^{r}$ is being shifted by $t_{S}$. Figure \ref{fig:evolutions}-c
presents this comparison, allowing us to see how the difference between
the $G^{r}$ and $S^{r}$ functions decays with the time.

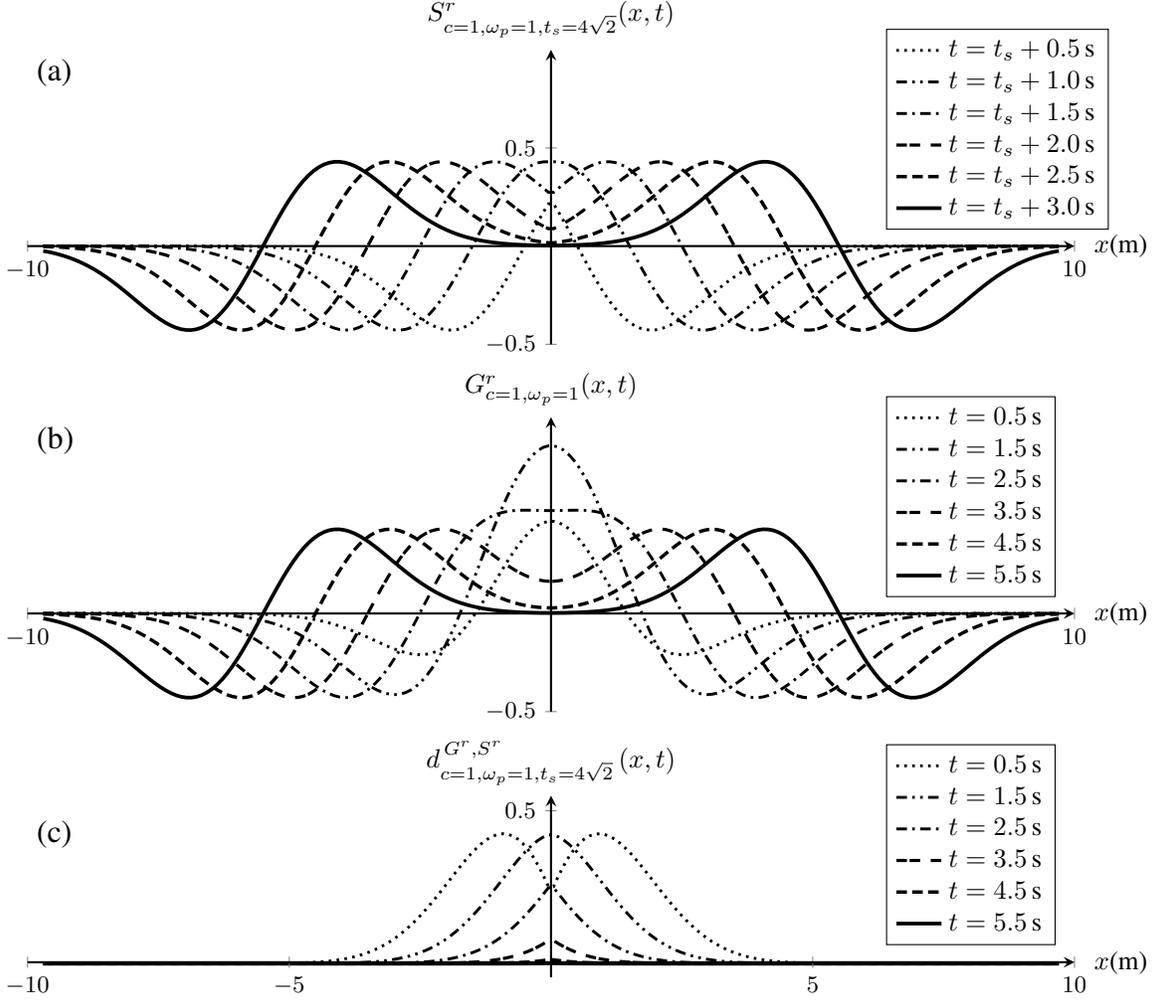
\begin{figure}
\begin{centering}
\input{figures/Sevolution.tex}
\par\end{centering}
\begin{centering}
\input{figures/Gevolution.tex}
\par\end{centering}
\begin{centering}
\input{figures/GmSevolution.tex}
\par\end{centering}
\caption{(a) Evolution of $S^{r}$ for $\sourceposition=0$, $c=1\,\mbox{m/s}$,
$\peakfreq=1\,\mbox{rad/s}$ and $\sourcetime=4\,\stdgauss=4\sqrt{2}\,\mbox{s}$.
(b) Evolution of $G$ for $\sourceposition=0$, $c=1$ m/s, $\peakfreq=1\,\mbox{rad/s}$.
(c) Evolution of the distance operator $d_{c,\peakfreq,\sourcetime}^{\,G^{r},S^{r}}\left(x,t\right)$
for $\sourceposition=0$, $c=1$ m/s, $\peakfreq=1\,\mbox{rad/s}$
and $\sourcetime=4\,\stdgauss=4\sqrt{2}\,\mbox{s}$.\label{fig:evolutions}}
\end{figure}

In order to better observe how the difference between $G^{r}$ and
$S^{r}$ depends on $t$ and is affected by the chosen $\sourcetime$,
we are now going to define a distance measure operation between $G^{r}$
and $S^{r}$. First, let the following norm of an absolutely integrable
function dependent on $x$ and $t$ be defined:
\begin{equation}
\left\Vert f\left(x,t\right)\right\Vert \left(t\right)=\int_{-\infty}^{\infty}\left|f\left(x,t\right)\right|dx,\;\;\;\;\;\;\;\;\;f\in L^{1}\left(\mathbb{R},x\right)\label{eq:norm}
\end{equation}
where $L^{1}\left(\mathbb{R},x\right)$ denotes the space of functions
that are absolutely integrable in $x\in\mathbb{R}$ $\left(\intop_{-\infty}^{\infty}\left|f(x,t)\right|dx<\infty\right)$.
Figure \ref{fig:GmSintegral} shows the numerical results of calculating
the distance operator 
\begin{equation}
D_{c,\peakfreq,\sourcetime}^{G^{r},S^{r}}\left(t\right)=\frac{\left\Vert d_{c,\peakfreq,\sourcetime}^{\,G^{r},S^{r}}\left(x,t\right)\right\Vert \left(t\right)}{\left\Vert G^{r}\left(x,t\right)\right\Vert \left(t\right)}\label{eq:distanceMeasure}
\end{equation}
where the values of $\sourcetime$ are in terms of the standard deviation
$\sigma$ of the corresponding gaussian function. The integrals were
computed only over the interval from $\sourceposition-ct-\bandwidth/2$
to $\sourceposition+ct+\bandwidth/2$, since outside it the values
of $G^{r}$ and $S^{r}$ are, by the Causality Principle and the definition
of $\bandwidth$, negligible. The values of $c$ and $\peakfreq$
were set as 1, with their respective units. As one can see, to the extent that we delay the source time $\sourcetime$, the distance between the solutions decays to smaller values as the time advances. However, we have to keep in mind that more delayed $\sourcetime$'s result in larger time lags between the $G$ and $S$ solutions, what can lead to major discrepancies when dealing with heterogeneous media.

\begin{figure}
\begin{centering}
\input{figures/GmSintegral.tex}
\par\end{centering}
\caption{Semi-log plot of $D_{c,\peakfreq,\sourcetime}^{G^{r},S^{r}}\left(t\right)$,
as defined by \eqref{eq:distanceMeasure}, with $c=1$ m/s and $\peakfreq=1$
rad/s, versus $t$, to measure how the distance between the $G^{r}$
and $S^{r}$ solutions decays as the time advances, according to distinct
source times $\sourcetime$'s.\label{fig:GmSintegral}}
\end{figure}
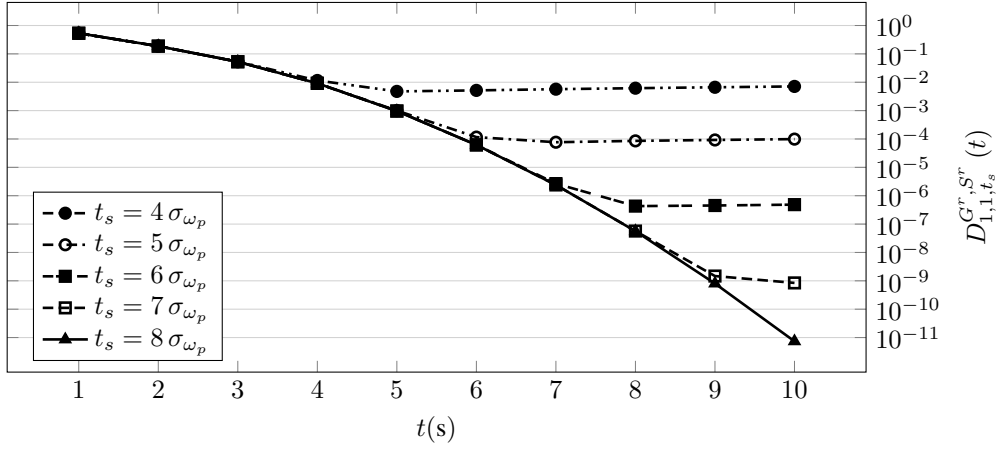

\section{Energy analysis of the solutions based on the Ricker wavelet \label{sec:Energy-Analysis}}

To a better understanding of how the solutions $G^{\psi}$ and $S^{\psi}$
form the wavefield, it is convenient to make an energy analysis of
the both solutions. The wave energy can be simply defined as follows:

\begin{equation}
E_{c}\left(t\right)=\cfrac{1}{2}\int_{-\infty}^{\infty}\left(\partial_{t}u\right)^{2}dx+\cfrac{c^{2}}{2}\int_{-\infty}^{\infty}\left(\partial_{x}u\right)^{2}dx\label{eq:energyexpr}
\end{equation}

Physically, $E_{c}(t)$ is the total energy of the system. The term
with $\left(\partial_{t}u\right)^{2}$ corresponds to the kinectic
energy, while the one with $\left(\partial_{x}u\right)^{2}$ corresponds
to the potential energy.

\begin{equation}
K\left(t\right)=\cfrac{1}{2}\int_{-\infty}^{\infty}\left(\partial_{t}u\right)^{2}dx\label{eq:Kexpr}
\end{equation}
\begin{equation}
U_{c}\left(t\right)=\cfrac{c^{2}}{2}\int_{-\infty}^{\infty}\left(\partial_{x}u\right)^{2}dx\label{eq:Uexpr}
\end{equation}

The homogeneous wave equation \eqref{eq:homowaveequation} has no
terms that insert or dissipate energy. Therefore, it is to be expected
that the identity $\frac{d}{dt}E_{c}(t)=0$ be valid.

To get the energy expressions associated with the Ricker-based initial
derivative solution, one can just substitute $u$ in \eqref{eq:Kexpr}
and \eqref{eq:Uexpr} by $G^{r}$ \eqref{eq:Gricker}. The algebraic
calculation is straightforward and lengthy, and can be directly done
with the aid of a computer algebra system. The obtained kinetic energy
expression is 
\begin{equation}
K_{c,\peakfreq}^{G^{r}}\left(t\right)=\frac{1}{8\,c\,\peakfreq}\sqrt{\frac{\pi}{2}}e^{-\frac{1}{2}\peakfreq^{2}t^{2}}\left(\peakfreq^{4}t^{4}-6\peakfreq^{2}t^{2}+3+3e^{\frac{1}{2}\peakfreq^{2}t^{2}}\right)\label{eq:EkGr}
\end{equation}
and the potential energy
\begin{equation}
U_{c,\peakfreq}^{G^{r}}\left(t\right)=\frac{1}{8\,c\,\peakfreq}\sqrt{\frac{\pi}{2}}e^{-\frac{1}{2}\peakfreq^{2}t^{2}}\left(-\peakfreq^{4}t^{4}+6\peakfreq^{2}t^{2}-3+3e^{\frac{1}{2}\peakfreq^{2}t^{2}}\right)\label{eq:EpGr}
\end{equation}

One can easily check that both the expressions given in \eqref{eq:EkGr} and \eqref{eq:EpGr} tend to the same constant value as time tends to infinity. 
The constant total energy associated with the $G^{r}$ solution is the sum
of the kinectic \eqref{eq:EkGr} and potential \eqref{eq:EpGr} energies:
\begin{equation}
E_{c,\peakfreq}^{G^{r}}=\frac{3}{4\,c\,\peakfreq}\sqrt{\frac{\pi}{2}}\label{eq:ETGr}
\end{equation}

\begin{figure}[h]
\input{figures/Genergy.tex}\input{figures/Senergy.tex}

\caption{Evolution in time of the kinectic energy $K_{c,\peakfreq}$, potential
energy $U_{c,\peakfreq}$ and total energy $E_{c,\peakfreq}$for the
$G^{r}$ and $S^{r}$ solutions of the wave equation. There are shown
the cases in which the peak frequency $\peakfreq$ is set as $1\,\mbox{rad/s}$
and $2\,\mbox{rad/s}$. (a) $G^{r}$ energy for $c=1$ m/s. (b) $S^{r}$
energy for $c=1$ m/s and $\sourcetime=5$ s.\label{fig:GSenergy}}
\end{figure}
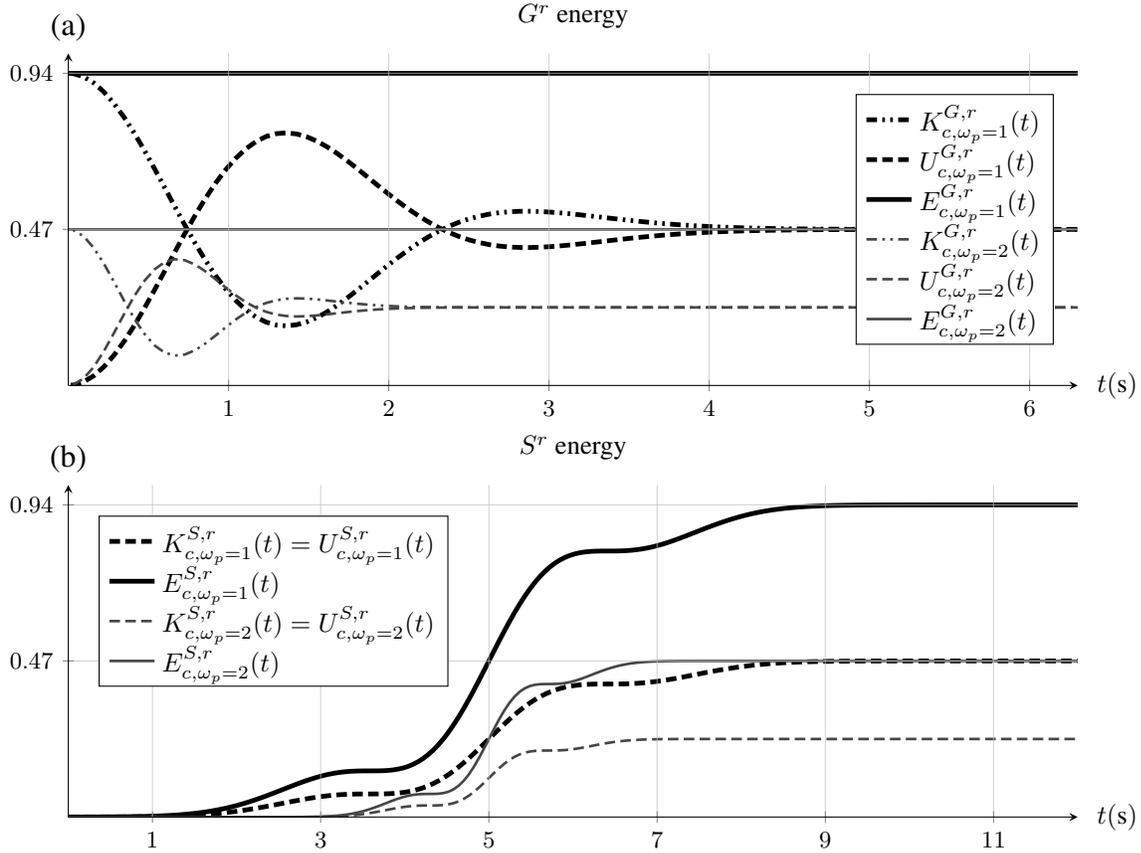

By applying the $S^{r}$ solution of the inhomogeneous wave equation
\eqref{eq:Sricker} to the energy formula \eqref{eq:energyexpr},
and after some tedious calculation, one can get the expression for
the kinetic energy: 
\begin{eqnarray}
K_{c,\peakfreq}^{S^{r}}\left(t\right) & = & \cfrac{3}{16\,c\,\peakfreq}\sqrt{\frac{\pi}{2}}\left[\text{erf}\left(\frac{\peakfreq\left(t-t_{s}\right)}{\sqrt{2}}\right)+1\right]\nonumber \\
 &  & -\cfrac{1}{16\,c}\left(t-t_{s}\right)e^{-\frac{1}{2}\peakfreq^{2}\left(t-t_{s}\right)^{2}}\left[\peakfreq^{2}\left(t-t_{s}\right)^{2}-1\right]\label{eq:EkSr}
\end{eqnarray}
where erf is the error function. Curiously, we get the same expression
for the potential energy $U_{c,\peakfreq}^{S^{r}}\left(t\right)$. This equipartition of energy in wave motion is predicted by Duffin\cite{duffin1970equipartition}, that has shown that, if the solution has compact support, in an odd-dimensional space, after a finite time, the kinetic energy of the wave is constant and equals the potential energy.
Therefore, the total energy has the expression:
\begin{eqnarray}
E_{c,\peakfreq}^{S^{r}}\left(t\right)=2K_{c,\peakfreq}^{S^{r}}\left(t\right) & = & \cfrac{3}{8\,c\,\peakfreq}\sqrt{\frac{\pi}{2}}\left[\text{erf}\left(\frac{\peakfreq\left(t-t_{s}\right)}{\sqrt{2}}\right)+1\right]\nonumber \\
 &  & -\cfrac{1}{8\,c}\left(t-t_{s}\right)e^{-\frac{1}{2}\peakfreq^{2}\left(t-t_{s}\right)^{2}}\left[\peakfreq^{2}\left(t-t_{s}\right)^{2}-1\right]\label{eq:ETSr}
\end{eqnarray}

Note that, since the point source inserts energy into the system,
the total energy associated with $S$ varies with $t.$ However, because of the short duration of the source, it can be easily shown that the maximum
value of \eqref{eq:ETSr}, corresponding to its limit as $t$ tends to
infinity, matches the same $G^{r}$ total energy \eqref{eq:ETGr} (appendix 
\ref{sec:LimitETSr}):
\begin{equation}
E_{c,\peakfreq,\textrm{max}}^{S^{r}}=\lim_{t\rightarrow\infty}E_{c,\peakfreq}^{S^{r}}\left(t\right)=\frac{3}{4\,c\,\peakfreq}\sqrt{\frac{\pi}{2}}=E_{c,\peakfreq}^{G^{r}}\label{eq:ETSrLimit}
\end{equation}

For example, for the case $c=1\,\mbox{m/s}$ and $\peakfreq=1\,\mbox{rad/s}$, we have $E_{c,\peakfreq,\textrm{max}}^{S^{r}}=E_{c,\peakfreq}^{G^{r}}\simeq0.94$.
In figure \ref{fig:GSenergy} is shown the energy evolution of each solution, with two different peak frequencies. In figure \ref{fig:GSenergy}-a, we can see that
all the energy of wave propagation for $G$ solution is available
from the initial time in the form of kinectic energy. We can also
observe that, after oscilating during a brief time, $K^{G,r}$ and
$U^{G,r}$ converge to the same value corresponding to the half of
the total energy. As shown in figure \ref{fig:GSenergy}-b, the IPS
releases energy to the system during a time determined by $\peakfreq$
and with a rate that reaches its maximum at $\sourcetime=5$ s, and
with $K^{G,r}$ being equal to $U^{G,r}$ all the time. The effect
of setting the peak frequency $\peakfreq$ as $2\,\mbox{rad/s}$ instead
of $1\,\mbox{rad/s}$ is to cut the total energy in half, as can be
predicted by the equation \eqref{eq:ETSrLimit}.

\section{Numerical Implementation\label{sec:Numerical-Implementation}}

A conventional way of discretizing the inhomogeneous wave equation
\eqref{eq:waveequation} consists of approximating the derivatives
of second order with centered differences in a mesh of grid points
$\left(x_{i},\,t_{n}\right)=\left(i\Delta x,\,n\Delta t\right)$ \cite{kukudzhanov2013numerical}:
\begin{equation}
\begin{cases}
\partial_{x}^{2}u\left(x_{i},t_{n}\right) & \approx\cfrac{u_{i+1,n}-2u_{i,n}+u_{i-1,n}}{\left(\Delta x\right)^{2}}\\
\partial_{t}^{2}u\left(x_{i},t_{n}\right) & \approx\cfrac{u_{i,n+1}-2u_{i,n}+u_{i,n-1}}{\left(\Delta t\right)^{2}}
\end{cases}\label{eq:numDeriv}
\end{equation}
where $i=1..\,N_{x}$ and $n=1..\,N_{t}$ are the spatial and temporal
indexes, respectively. When \eqref{eq:numDeriv} is applied to the
inhomogeneous wave equation \eqref{eq:waveequation} leads to the
finite difference (FD) scheme

\begin{equation}
u_{i,n+1}^{(s)}=-u_{i,n-1}^{(s)}+2u_{i,n}^{(s)}+C^{2}\left(u_{i+1,n}^{(s)}-2u_{i,n}^{(s)}+u_{i-1,n}^{(s)}\right)-\Delta t^{2}s_{i,n}\label{eq:numSolve}
\end{equation}
where the the superscript $(s)$ makes explicit that the energy of
this solution is provided by the source, and $C=c\frac{\Delta t}{\Delta x}$
is the so-called \textit{Courant number}, a dimensionless number that
is related to the numerical stability of solving the wave equation
by using the finite difference method. We consider the simple case
in which $c$ (and consequently also $C$) is uniform in all the mesh,
that is, the medium is homogeneous (not to be confused with the homogeneous
wave equation) The Courant\textendash Friedrichs\textendash Lewy (CFL)
condition demands that, for one-dimensional case \cite{courant1967partial},
\begin{equation}
C\leq1\label{eq:cflcond}
\end{equation}

The expression \eqref{eq:numSolve} provides an iterative procedure
for solving the wave equation. Since the time-related index $n$ is
iterated, \eqref{eq:numSolve} represents a time-domain FD method.

An Impulsive Point Source (IPS) can be approximated in the discrete domain
in the following way:
\begin{equation}
s_{i,n}=\psi_{n}\delta_{i,i_{s}}\label{eq:si}
\end{equation}
in which $\psi_{n}$ is the discrete sampling, at time index $n$,
of the continuous wavelet $\psi$ centered at the time $t_{s}$, and
$\delta_{i,i_{s}}$ represents the Kronecker delta, which is equal
to 1 only when $i=i_{s}$, whith $i_{s}$ being the spatial index
that localizes the source.

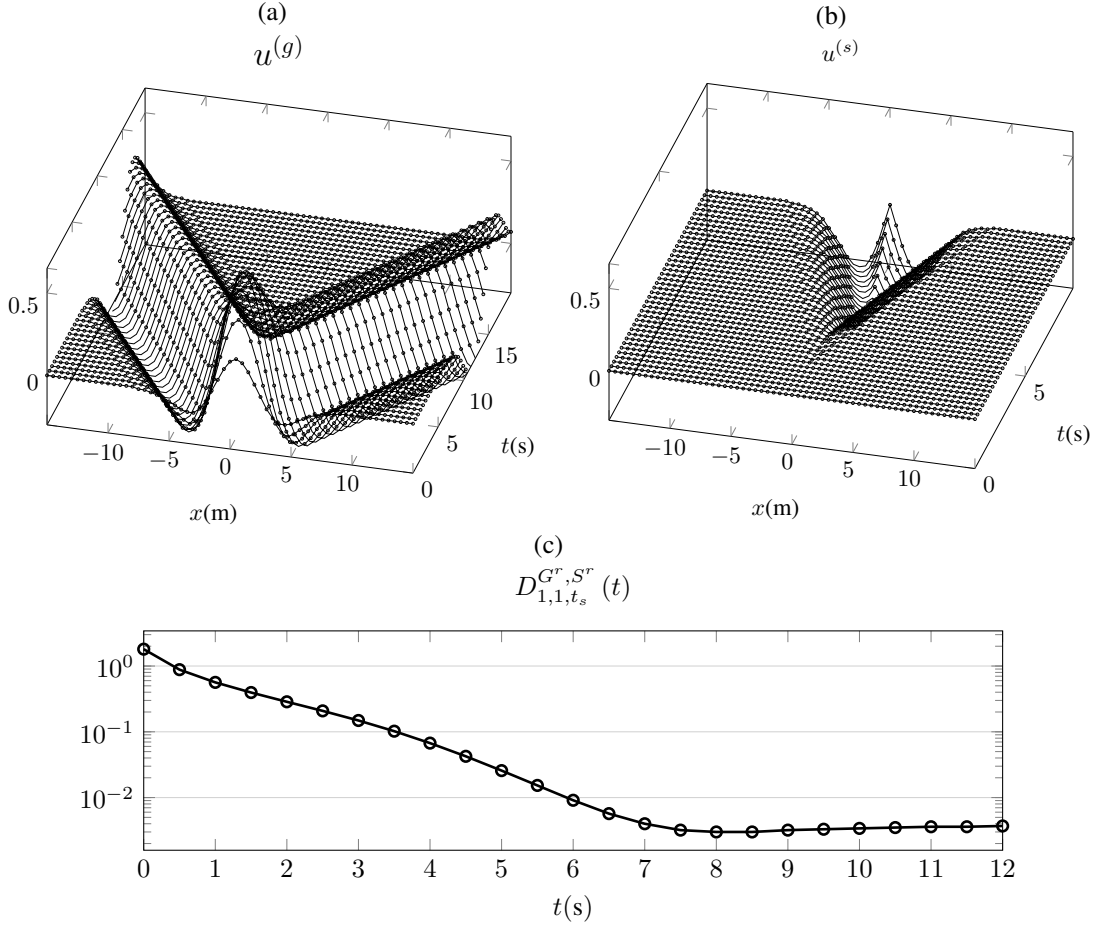
\begin{figure}[h]
\begin{centering}
\input{figures/GSnumData.tex}
\par\end{centering}
\caption{Numerical simulation of wave propagation by solving the wave equation
with the finite difference method. (a) Setting the initial time derivative
condition, without a source. (b) Setting an impulsive point source (IPS).
The parameters are: $\sourceposition=0,\,\sourcetime=10\,\mbox{s},\,c=1\,\mbox{m/s},\,\omega_{p}=1\,\mbox{rad/s},\,\Delta x=0.5\,\mbox{m},\,\Delta t=0.5\,\mbox{s}.$
\label{fig:GnumData}}
\end{figure}

The time derivative in the initial conditions \eqref{eq:initialconditions}
can be embedded in a simple way by taking the following approximation:
\begin{equation}
\left.\partial_{t}u\left(x_{i},t_{n}\right)\right|_{n=0}=g_{i}\approx\cfrac{u_{i,1}-u_{i,-1}}{2\Delta t}\label{eq:derivativeAproximation}
\end{equation}

Inserting \eqref{eq:derivativeAproximation} into \eqref{eq:numSolve},
with $n=0$, and setting $s_{i,0}=0$ at every point of the grid (that
is, there is no source acting at the initial time), one gets

\begin{equation}
u_{i,1}=u_{i,0}+\Delta t\,g_{i}+\cfrac{1}{2}\,C^{2}\left(u_{i+1,0}^ {}-2u_{i,0}^ {}+u_{i-1,0}^ {}\right)\label{eq:numFirst}
\end{equation}
which corresponds to the first iteration of time-domain FD scheme.
If we consider $u_{i}^{0}=0$ at every point of the grid (that is,
no initial wavefield), the expression is reduced to:
\begin{equation}
u_{i,1}^{(g)}=\Delta t\,g_{i}\label{eq:numFirstgi}
\end{equation}
where the superscript $(g)$ makes explicit that the energy is available
from the beginning of the simulation via the time derivative initial
condition. Seeing that the computational domain is limited, one has
to design artificial boundaries satisfying the chosen boundary conditions.
We will not discuss here the artificial boundaries methods, since
they are irrelevant to the main point of this work if the spatial
domain is large enough.

The present work states that one can substitute a wavelet-based IPS
$s_{i,n}$ by setting $g_{i}$ in \eqref{eq:numFirst} as the sampled
wavelet:
\begin{equation}
g_{i}=\cfrac{1}{c}\,\psi_{\left(i-i_{s}\right)/c}\label{eq:gi}
\end{equation}
and removing the $s$ term in \eqref{eq:numSolve}
\begin{equation}
u_{i,n+1}^{(g)}=-u_{i,n-1}^{(g)}+2u_{i,n}^{(g)}+C^{2}\left(u_{i+1,n}^{(g)}-2u_{i,n}^{(g)}+u_{i-1,n}^{(g)}\right)\label{eq:numSolveHomogeneous}
\end{equation}

Figure \ref{fig:GnumData} presents the results of two simulations
by using the described methodologies based on $s$ and $g$, with
$c=1\,\mbox{m/s}$ and $\peakfreq=1\,\mbox{rad/s}$. Although the
space-time domain is larger, the plot shows a set of points corresponding
to $N_{x}=61$ discrete positions, varying from $x=-15\,\mbox{m}$
to $x=15\,\mbox{m}$ ($\Delta x=0.5\,\mbox{m}$), and $N_{t}=41$
discrete times, varying from $t_{0}=0\,\mbox{s}$ to $t=20\,\mbox{s}$
($\Delta t=0.5\,\mbox{s}$). With this set of parameters, the Courant
Number $C$ equals 1, satisfying the CFL condition \eqref{eq:cflcond}
in the limit. Figure \ref{fig:GnumData}-a shows the evolution of
the wavefield when one sets $g_{i}$ in the first iteration \eqref{eq:numFirstgi}
as the sampled wavelet \eqref{eq:gi}, and use \eqref{eq:numSolveHomogeneous}
as the FD scheme. Figure \ref{fig:GnumData}-b shows the evolution
of the wavefield when both $u_{i,0}$ and $u_{i,1}$ equal 0 at every
point of the grid, and one uses \eqref{eq:numSolve} as the FD scheme,
with $s_{i}$ being defined by \eqref{eq:si}. The source is located
at $\sourceposition=0$ and reaches its maximum energy transfer at
$\sourcetime=10\,\mbox{s}$.

To compare these two results, we implement a discrete version of the
distance operator \eqref{eq:distanceMeasure}:
\begin{equation}
D_{n}^{G,S}=\cfrac{\sum_{i=1}^{N_{x}}\left|u_{i,n}^{(g)}-u_{i,n+n_{s}}^{(s)}\right|}{\sum_{i=1}^{N_{x}}\left|u_{i,n}^{(g)}\right|}\label{eq:discreteDistanceOperator}
\end{equation}
where $n_{s}$ is the time index associated with $\sourcetime$. Since
$\sourcetime=10\,\mbox{s}$ and $\Delta t=0.5\,\mbox{s}$, we have
for this case $n_{s}=20$. Figure \ref{fig:GnumData}-c shows the
evolution of $D_{n}^{G,S}$, demonstrating that the waveform $u^{(g)}$
comes to have approximately the same shape of $u^{(s)}$ after about
8 seconds, despite being out of phase, as indicated by the distance
measure below $10^{-2}$ ($<1\%$). This is not so persuasive as the
theoretical predicition expressed in the figure \ref{fig:GmSintegral},
what we credit to the imperfections of the discretizations \eqref{eq:numDeriv}
and \eqref{eq:derivativeAproximation}, but it is still a good clue
that the IPS can be numerically emulated by setting a suitable initial
time derivative condition, and so slightly reducing the computational
cost of the modelling process.

\section{Conclusions\label{sec:Conclusions}}

In time-domain modeling of the wave equation for seismic applications,
the wave propagation is usually triggered by setting a wavelet-based
source term. We have shown that, by using a scaled wavelet as the
initial time derivative condition of the one-dimensional homogeneous
wave equation, one can generate the same waveforms that would be generated
by an impulsive point source. So, in these specific conditions, the inhomogeneous
version of the wave equation can be avoided by suppressing its source
term. This was corroborated by numerical results performed in a homogeneous
media model.

The use of this initial condition technique has some limitations.
It produces wavefronts with a time lag with respect to the ones generated
by the IPS's, if we consider the time required for the source operation.
Furthermore, it instantly affects all the spatial range covered by
the wavelet, while the IPS directly operates only on a point.
In non-homogeneous media, where the velocity of wave propagation varies
with the position, this can lead to the formation of slightly different
waveforms between the two methodologies.

We believe that this study can be expanded to the 2D and 3D cases, what would make it suitable for practical applications such as wave propagation modeling for seismic imaging. However, one has to note that the wavelet that sets up the initial time derivative condition would have the same dimensions as the problem, while the wavelet that sets up the source term would continue to be a one-dimensional function. For this expansion to be made analytically, one would have to deal with the Kirchoff's formula for the solution of the wave equation in $\mathbb{R}^{2}$ and $\mathbb{R}^{3}$.

\section{Acknowledgments}

The author wish to dedicate this work to the memory of prof. Liacir dos Santos Lucena, and gratefully acknowledge the support of the Universidade Federal do Rio Grande do Norte (UFRN) and Universidade Federal Rural do Semi-Árido (UFERSA).

\appendix

\section{Obtaining the $S^{\psi}$ solution\label{sec:derivS}}

Here, we aim to obtain the $S^{\psi}$ solution when one
uses an Impulsive Point Source (IPS) as described by the definition \eqref{eq:single-pulse-source}. From the expression \eqref{eq:Sinit}, since $\psi(t'-t_{0})$ does not depend on
$x'$:
\[
S_{c,\sourceposition,\sourcetime}^{\psi}\left(x,t\right)=\cfrac{1}{2c}\int_{0}^{t}\psi\left(t'-\sourcetime\right)\int_{x-c\left(t-t'\right)}^{x+c\left(t-t'\right)}\delta\left(x'-\sourceposition\right)dx'dt'
\]

Let us make the transformation $r=c\left(t-t'\right)$, so that $t'=\left(r+ct\right)/c=r/c+t$,
$dt'=-dr/c$, and the interval of integration $0\rightarrow t$ becomes
$ct\rightarrow0$:
\[
S_{c,\sourceposition,\sourcetime}^{\psi}\left(x,t\right)=-\cfrac{1}{2c^{2}}\int_{ct}^{0}\psi\left(r/c+t-\sourcetime\right)\int_{x-r}^{x+r}\delta\left(x'-\sourceposition\right)dx'dr
\]

Let $I_{\sourceposition}(x,r)$ be the result of the integral in $x'$,
which, taking into account that $r\geq0$, can be defined by the boxcar
function:
\[
I_{\sourceposition}(x,r)=\int_{x-r}^{x+r}\delta\left(x'-\sourceposition\right)dx'=\Pi_{x-r,\,x+r}\left(\sourceposition\right)
\]
which, in turn, can be written in terms of the Heaviside Step Function:
\[
I_{\sourceposition}(x,r)=H_{x-r}\left(\sourceposition\right)-H_{x+r}\left(\sourceposition\right)
\]
where $H_{a}\left(x\right)=H\left(x-a\right)$. Rewriting it in function
of $r$:
\[
I_{\sourceposition}(x,r)=H_{x-\sourceposition}\left(r\right)-H_{x-\sourceposition}\left(-r\right)
\]

One can verify that:


\begin{equation}
H_{x-\sourceposition}\left(r\right)-H_{x-\sourceposition}\left(-r\right)=H_{\left|x-\sourceposition\right|}\left(r\right)-H_{\left|x-\sourceposition\right|}\left(-r\right)\label{eq:heaviside_relation}
\end{equation}

The non-null values of the last term in \eqref{eq:heaviside_relation}
are out of the interval $0\leq r\leq ct$, and, therefore, we can
consider:
\[
I_{\sourceposition}(x,r)=H_{\left|x-\sourceposition\right|}\left(r\right)
\]
which is equivalent to the following condition:
\begin{equation}
I_{\sourceposition}(x,r)=\begin{cases}
1, & \mbox{if }r>\left|x-\sourceposition\right|\\
1/2 & \mbox{if }r=\left|x-\sourceposition\right|\\
0, & \mbox{otherwise}
\end{cases}\label{eq:ucondition}
\end{equation}

Coming back to the variable $t$':
\begin{eqnarray*}
I_{\sourceposition}(x,r) & = & H_{\left|x-\sourceposition\right|}\left[c\left(t-t'\right)\right]\\
 & = & H\left[c\left(t-t'-\frac{\left|x-\sourceposition\right|}{c}\right)\right]
\end{eqnarray*}

By using the identity $H\left(ax\right)=H\left(x\right)H\left(a\right)+H\left(-x\right)H\left(-a\right)$,
we obtain:
\[
I_{\sourceposition}(x,r)=H\left(c\right)H\left(t-t'-\cfrac{\left|x-\sourceposition\right|}{c}\right)+H\left(-c\right)H\left(-t+t'+\cfrac{\left|x-\sourceposition\right|}{c}\right)
\]

But, considering that $c\geq0$, we have $H\left(c\right)=1$ and
$H\left(-c\right)=0$:
\begin{eqnarray*}
I\left(r\right) & = & H\left(t-t'-\cfrac{\left|x-\sourceposition\right|}{c}\right)\\
 & = & H_{-\left(t-\left|x-\sourceposition\right|/c\right)}\left(-t'\right)
\end{eqnarray*}

Thus, we arrive at:
\begin{eqnarray*}
S_{c,\sourceposition,\sourcetime}^{\psi}\left(x,t\right) & = & -\cfrac{1}{2c^{2}}\int_{ct}^{0}\psi\left(r/c+t-\sourcetime\right)I_{\sourceposition}(x,r)dr\\
 & = & \cfrac{1}{2c}\int_{0}^{t}\psi\left(t'-\sourcetime\right)H_{-\left(t-\left|x-\sourceposition\right|/c\right)}\left(-t'\right)dt'
\end{eqnarray*}

The function $H_{-a}\left(-x\right)$, for $a\geq0$, is non-null
only for $x\leq a$, and, therefore, the interval of integration becomes:
\[
S_{c,\sourceposition,\sourcetime}^{\psi}\left(x,t\right)=\cfrac{1}{2c}\int_{0}^{t-\left|x-\sourceposition\right|/c}\psi\left(t'-\sourcetime\right)dt'
\]
\textcolor{blue}{}with $\varphi\left(x\right)$ being the antiderivative
of $\psi\left(x\right)$.\textcolor{blue}{{} }By applying the fundamental
theorem of calculus, we then get the expression \eqref{eq:Ssolution}.\textcolor{red}{}

\section{Obtaining the $G^{\psi}$ solution\label{sec:derivG}}

From \eqref{eq:Gdefinition}:
\[
G_{c,\sourceposition}^{\psi}\left(x,t\right)=\cfrac{1}{2c^{2}}\int_{x-ct}^{x+ct}\psi\left(\cfrac{x'-\sourceposition}{c}\right)dx'
\]

Let $\eta=\cfrac{x'-\sourceposition}{c}$, so that $dx'=c\,d\eta$:
\[
G_{c,\sourceposition}^{\psi}\left(x,t\right)=\cfrac{1}{2c}\int_{\left(x-\sourceposition\right)/c-t}^{\left(x-\sourceposition\right)/c+t}\psi\left(\eta\right)d\eta
\]

Since $\varphi$ is the antiderivative of $\psi$, by the fundamental
theorem of calculus we get solution \eqref{eq:Gsolution}:
\[
G_{c,\sourceposition}^{\psi}\left(x,t\right)=\cfrac{1}{2c}\left[\varphi\left(\cfrac{x-\sourceposition}{c}+t\right)-\varphi\left(\cfrac{x-\sourceposition}{c}-t\right)\right]
\]

\section{Calculus of $\stdgauss$\label{sec:CalculusStdGauss}}

In this appendix, we derive the standard deviation $\stdgauss$ of
the gaussian function $f_{g}$ that corresponds to the antiderivative
of $\phirickeromega$. Since both $f_{g}$ and $\phirickeromega$
have zero mean (no constants are added):
\[
f_{g}\left(t\right)=\int\phirickeromega\left(t\right)dt=\int t\,e^{-\omega_{p}^{2}t^{2}/4}dt=-\cfrac{2e^{-\peakfreq^{2}t^{2}/4}}{\peakfreq^{2}}
\]

By comparing this with the general form of the gaussian function centered
at zero $ae^{-x^{2}/2\sigma^{2}}$, where $a$ is an amplitude factor
and $\sigma$ is its standard deviation, we arrive at the conclusion
that 
\[
\stdgauss=\sqrt{2}/\peakfreq
\]

\section{Limit of $E^{S^{r}}$\label{sec:LimitETSr}}

We want to prove the \eqref{eq:ETSrLimit} relation, correponding
to the maximum value of the total energy of the $S^{r}$ solution,
which approaches its limit as $t$ tends to infinity.

By expanding \eqref{eq:ETSr}, we get
\[
E_{c,\peakfreq}^{S^{r}}\left(t\right)=\cfrac{3}{8\,c\,\peakfreq}\sqrt{\frac{\pi}{2}}\text{erf}\left(\frac{\peakfreq t}{\sqrt{2}}\right)+\frac{3}{8\,c\,\peakfreq}\sqrt{\frac{\pi}{2}}+\frac{te^{-\frac{1}{2}\peakfreq^{2}t^{2}}}{8\,c}-\frac{\peakfreq^{2}t^{3}e^{-\frac{1}{2}\peakfreq^{2}t^{2}}}{8\,c}
\]

Since the gaussian function $e^{-\frac{1}{2}\peakfreq^{2}t^{2}}$
is rapidly decreasing, the last two terms tend to zero as $t\rightarrow\infty$.
The erf function tends to 1, so:
\[
\lim_{t\rightarrow\infty}E_{c,\peakfreq}^{S^{r}}\left(t\right)=\frac{3}{8\,c\,\peakfreq}\sqrt{\frac{\pi}{2}}+\frac{3}{8\,c\,\peakfreq}\sqrt{\frac{\pi}{2}}=\frac{3}{4\,c\,\peakfreq}\sqrt{\frac{\pi}{2}}
\]
validating \eqref{eq:ETSrLimit}.

\nocite{*}
\bibliographystyle{unsrtnat}

\bibliography{references}

\end{document}

%% file: figures/rickerwavelet.tex

\begin{tikzpicture}[scale=0.8]
	\def\AA{1}
	\def\omegaA{1}

	\pgfplotsset{every tick label/.append style={font=\small}} 
	\begin{axis}[
		title={(a)},
		xmin=-5, xmax=5,
		ymin=-1, ymax=1.2,
		domain=-5:4.7,
		xtick={-3,-1,1,3},
		ytick={-1,1},
		samples=200, 
		legend cell align=left,
		legend style={at={(0.70,0.95)},anchor=west},
		axis on top=true,
		axis x line=middle,
		axis y line=middle,
		inner axis line style={thick},
		xlabel=$t$(s),
		every axis x label/.style={at={(ticklabel* cs:1.01)}, anchor=west},
	]
	\addplot[line width=1.1pt] (x,{\AA*(1-(\omegaA^2*x^2)/2)*exp(-(\omegaA^2*x^2)/4)});
	\addplot[line width=1.1pt, dash pattern=on 4pt off 1pt] (x,{x*exp(-(\omegaA^2*x^2)/4)});

    \legend{$\psirickeromegadef{1}(t)$,$\phirickeromegadef{1}(t)$} 
	\end{axis}
\end{tikzpicture}\;
\begin{tikzpicture}[scale=0.8]
	\def\AA{1}
	\def\omegaA{1}
	\def\omegaB{0.5}
	\def\omegaC{0.25}

	\pgfplotsset{every tick label/.append style={font=\small}} 
	\begin{axis}[
		title={(b)},
		ymin=-3.6, ymax=3.6,
		xmin=-20, xmax=20,
		domain=-20:19,
		ytick={-3,1,3},
		samples=200, 
		legend cell align=left,
		legend style={at={(0.55,0.20)},anchor=west},
		axis on top=true,
		axis x line=middle,
		axis y line=middle,
		inner axis line style={thick},
		xlabel=$t$(s),
		every axis x label/.style={at={(ticklabel* cs:1.01)}, anchor=west},
	]
	\addplot[line width=1.1pt] 	(x,	{x*exp(-(\omegaA^2*x^2)/4)}	);
	\addplot[line width=1.1pt, dash pattern=on 4pt off 1pt] 	(x,   {x*exp(-(\omegaB^2*x^2)/4)}	);
	\addplot[line width=1.1pt, dash pattern=on 4pt off 2pt on 1pt off 2pt] 	(x,   {x*exp(-(\omegaC^2*x^2)/4)}	);



	\legend{$\phirickeromegadef{1}(t)$,$\phirickeromegadef{1/2}(t)$,$\phirickeromegadef{1/4}(t)$} 
	\end{axis}
\end{tikzpicture}

%% file: figures/Sevolution.tex

\begin{tikzpicture}[scale=1.0]

	\def\xm{10}
	\def\wp{1.0}
	\def\cc{1.0}

	\def\tO{5.65685}
	\def\tA{6.15685}
	\def\tB{7.15685}
	\def\tC{8.15685}
	\def\tD{9.15685}
	\def\tE{10.15685}
	\def\tF{11.15685}

	\pgfplotsset{every tick label/.append style={font=\small}} 
	\begin{axis}[
		width=15.5cm,
		height=5.5cm,
		xmin=-\xm, xmax=\xm,
		ymin=-0.5, ymax=1.0,
		domain=-\xm+0.3:\xm-0.3,
		xtick={-\xm,\xm},
		ytick={-0.5, 0.5},
		samples=200, 
		legend cell align=left,
		legend style={at={(0.82,0.73)},anchor=west},
		axis on top=true,
		axis x line=middle,
		axis y line=middle,
		inner axis line style={thick},
		xlabel=$x$(m),
		ylabel={$S^{r}_{c=1,\peakfreq=1,\sourcetime=4\sqrt{2}}(x,t)$},
		every axis y label/.style={at={(ticklabel* cs:1.01)}, anchor=south}, 
		every axis x label/.style={at={(ticklabel* cs:1.01)}, anchor=west}
	]
	\addplot[line width=1.1pt, dash pattern=on 1pt off 2pt] (x,{(\cc/2)*(\tO*exp(-((\wp^2)/4)*(\tO^2)) + (\tA-\tO-abs(x))*exp(-((\wp^2)/4)*(-\tA+\tO+abs(x))^2))});
	\addplot[line width=1.1pt, dash pattern=on 4pt off 2pt on 1pt off 2pt on 1pt off 2pt] (x,{(\cc/2)*(\tO*exp(-((\wp^2)/4)*(\tO^2)) + (\tB-\tO-abs(x))*exp(-((\wp^2)/4)*(-\tB+\tO+abs(x))^2))});
	\addplot[line width=1.1pt, dash pattern=on 4pt off 2pt on 1pt off 2pt] (x,{(\cc/2)*(\tO*exp(-((\wp^2)/4)*(\tO^2)) + (\tC-\tO-abs(x))*exp(-((\wp^2)/4)*(-\tC+\tO+abs(x))^2))});
	\addplot[line width=1.2pt,  dash pattern=on 4pt off 1pt on 4pt off 4pt] (x,{(\cc/2)*(\tO*exp(-((\wp^2)/4)*(\tO^2)) + (\tD-\tO-abs(x))*exp(-((\wp^2)/4)*(-\tD+\tO+abs(x))^2))});
	\addplot[line width=1.3pt, dash pattern=on 4pt off 2pt] (x,{(\cc/2)*(\tO*exp(-((\wp^2)/4)*(\tO^2)) + (\tE-\tO-abs(x))*exp(-((\wp^2)/4)*(-\tE+\tO+abs(x))^2))});
	\addplot[line width=1.4pt] (x,{(\cc/2)*(\tO*exp(-((\wp^2)/4)*(\tO^2)) + (\tF-\tO-abs(x))*exp(-((\wp^2)/4)*(-\tF+\tO+abs(x))^2))});

    \legend{
		$t=\sourcetime+0.5\,\mbox{s}$,
		$t=\sourcetime+1.0\,\mbox{s}$,
		$t=\sourcetime+1.5\,\mbox{s}$,
		$t=\sourcetime+2.0\,\mbox{s}$,
		$t=\sourcetime+2.5\,\mbox{s}$,
		$t=\sourcetime+3.0\,\mbox{s}$,
	} 

	\node at (axis cs:-10,1) [anchor=north west] {\large (a)}; 
	\end{axis}
\end{tikzpicture}

%% file: figures/Gevolution.tex

\begin{tikzpicture}[scale=1.0]

	\def\xm{10}
	\def\wp{1.0}
	\def\cc{1.0}

	\def\tA{0.5}
	\def\tB{1.5}
	\def\tC{2.5}
	\def\tD{3.5}
	\def\tE{4.5}
	\def\tF{5.5}

	\pgfplotsset{every tick label/.append style={font=\small}} 
	\begin{axis}[
		width=15.5cm,
		height=5.5cm,
		xmin=-\xm, xmax=\xm,
		ymin=-0.5, ymax=1.0,
		domain=-\xm+0.3:\xm-0.3,
		xtick={-\xm,\xm},
		ytick={-0.5},
		samples=200, 
		legend cell align=left,
		legend style={at={(0.82,0.73)},anchor=west},
		axis on top=true,
		axis x line=middle,
		axis y line=middle,
		inner axis line style={thick},
		xlabel=$x$(m),
		ylabel={$G^{r}_{c=1,\peakfreq=1}(x,t)$},
		every axis y label/.style={at={(ticklabel* cs:1.01)}, anchor=south}, 
		every axis x label/.style={at={(ticklabel* cs:1.01)}, anchor=west}
	]
	\addplot[line width=1.1pt, dash pattern=on 1pt off 2pt] (x,{(\cc/2)*((\tA-x)*exp(-((\wp^2)/4)*(\tA-x)^2) + (\tA+x)*exp(-((\wp^2)/4)*(\tA+x)^2))});
	\addplot[line width=1.1pt, dash pattern=on 4pt off 2pt on 1pt off 2pt on 1pt off 2pt] (x,{(\cc/2)*((\tB-x)*exp(-((\wp^2)/4)*(\tB-x)^2) + (\tB+x)*exp(-((\wp^2)/4)*(\tB+x)^2))});
	\addplot[line width=1.1pt, dash pattern=on 4pt off 2pt on 1pt off 2pt] (x,{(\cc/2)*((\tC-x)*exp(-((\wp^2)/4)*(\tC-x)^2) + (\tC+x)*exp(-((\wp^2)/4)*(\tC+x)^2))});
	\addplot[line width=1.2pt, dash pattern=on 4pt off 1pt on 4pt off 4pt] (x,{(\cc/2)*((\tD-x)*exp(-((\wp^2)/4)*(\tD-x)^2) + (\tD+x)*exp(-((\wp^2)/4)*(\tD+x)^2))});
	\addplot[line width=1.3pt, dash pattern=on 4pt off 2pt] (x,{(\cc/2)*((\tE-x)*exp(-((\wp^2)/4)*(\tE-x)^2) + (\tE+x)*exp(-((\wp^2)/4)*(\tE+x)^2))});
	\addplot[line width=1.4pt] (x,{(\cc/2)*((\tF-x)*exp(-((\wp^2)/4)*(\tF-x)^2) + (\tF+x)*exp(-((\wp^2)/4)*(\tF+x)^2))});

    \legend{
		$t=\tA\,$s,
		$t=\tB\,$s,
		$t=\tC\,$s,
		$t=\tD\,$s,
		$t=\tE\,$s,
		$t=\tF\,$s,
	} 

	\node at (axis cs:-10,1) [anchor=north west] {\large (b)}; 
	\end{axis}
\end{tikzpicture}

%% file: figures/GmSevolution.tex

\begin{tikzpicture}[scale=1.0]

	\def\xm{10}
	\def\wp{1.0}
	\def\cc{1.0}

	\def\tO{5.0}

	\def\tA{0.5}
	\def\tB{1.5}
	\def\tC{2.5}
	\def\tD{3.5}
	\def\tE{4.5}
	\def\tF{5.5}

	\pgfplotsset{every tick label/.append style={font=\small}} 
	\begin{axis}[
		width=15.5cm,
		height=4cm,
		xmin=-\xm, xmax=\xm,
		ymin=-0.05, ymax=0.55,
		domain=-\xm+0.3:\xm-0.3,
		xtick={-10,-5,0,5,10},
		ytick={0.5},
		samples=200, 
		legend cell align=left,
		legend style={at={(0.82,0.73)},anchor=west},
		axis on top=true,
		axis x line=middle,
		axis y line=middle,
		inner axis line style={thick},
		xlabel=$x$(m),
		ylabel={$d_{c=1,\peakfreq=1,\sourcetime=4\sqrt{2}}^{\,G^{r},S^{r}}\left(x,t\right)$},
		every axis y label/.style={at={(ticklabel* cs:1.01)}, anchor=south}, 
		every axis x label/.style={at={(ticklabel* cs:1.01)}, anchor=west}
	]
	\addplot[line width=1.1pt, dash pattern=on 1pt off 2pt] (x,{(\cc/2)*((\tA-x)*exp(-((\wp^2)/4)*(\tA-x)^2) + (\tA+x)*exp(-((\wp^2)/4)*(\tA+x)^2)) - (\cc/2)*(\tO*exp(-((\wp^2)/4)*(\tO^2)) + ((\tA+\tO)-\tO-abs(x))*exp(-((\wp^2)/4)*(-(\tA+\tO)+\tO+abs(x))^2))});
	\addplot[line width=1.1pt, dash pattern=on 4pt off 2pt on 1pt off 2pt on 1pt off 2pt] (x,{(\cc/2)*((\tB-x)*exp(-((\wp^2)/4)*(\tB-x)^2) + (\tB+x)*exp(-((\wp^2)/4)*(\tB+x)^2)) -  (\cc/2)*(\tO*exp(-((\wp^2)/4)*(\tO^2)) + ((\tB+\tO)-\tO-abs(x))*exp(-((\wp^2)/4)*(-(\tB+\tO)+\tO+abs(x))^2))});
	\addplot[line width=1.1pt, dash pattern=on 4pt off 2pt on 1pt off 2pt] (x,{(\cc/2)*((\tC-x)*exp(-((\wp^2)/4)*(\tC-x)^2) + (\tC+x)*exp(-((\wp^2)/4)*(\tC+x)^2)) - (\cc/2)*(\tO*exp(-((\wp^2)/4)*(\tO^2)) + ((\tC+\tO)-\tO-abs(x))*exp(-((\wp^2)/4)*(-(\tC+\tO)+\tO+abs(x))^2))});
	\addplot[line width=1.2pt, dash pattern=on 4pt off 1pt on 4pt off 4pt] (x,{(\cc/2)*((\tD-x)*exp(-((\wp^2)/4)*(\tD-x)^2) + (\tD+x)*exp(-((\wp^2)/4)*(\tD+x)^2)) - (\cc/2)*(\tO*exp(-((\wp^2)/4)*(\tO^2)) + ((\tD+\tO)-\tO-abs(x))*exp(-((\wp^2)/4)*(-(\tD+\tO)+\tO+abs(x))^2))});
	\addplot[line width=1.3pt, dash pattern=on 4pt off 2pt] (x,{(\cc/2)*((\tE-x)*exp(-((\wp^2)/4)*(\tE-x)^2) + (\tE+x)*exp(-((\wp^2)/4)*(\tE+x)^2)) - (\cc/2)*(\tO*exp(-((\wp^2)/4)*(\tO^2)) + ((\tE+\tO)-\tO-abs(x))*exp(-((\wp^2)/4)*(-(\tE+\tO)+\tO+abs(x))^2))});
	\addplot[line width=1.4pt] (x,{(\cc/2)*((\tF-x)*exp(-((\wp^2)/4)*(\tF-x)^2) + (\tF+x)*exp(-((\wp^2)/4)*(\tF+x)^2)) - (\cc/2)*(\tO*exp(-((\wp^2)/4)*(\tO^2)) + ((\tF+\tO)-\tO-abs(x))*exp(-((\wp^2)/4)*(-(\tF+\tO)+\tO+abs(x))^2))});

    \legend{
		$t=\tA\,$s,
		$t=\tB\,$s,
		$t=\tC\,$s,
		$t=\tD\,$s,
		$t=\tE\,$s,
		$t=\tF\,$s,
	}

	\node at (axis cs:-10,0.5) [anchor=north west] {\large (c)}; 
	\end{axis}
\end{tikzpicture}

%% file: figures/GmSintegral.tex

\begin{tikzpicture}[scale=1.0]

\begin{semilogyaxis}[
	width=13cm,
	height=6.5cm,
	xlabel=$t(\textrm{s})$,
	ylabel={$D^{G^{r},S^{r}}_{1,1,\sourcetime}\left(t\right)$},
	ylabel near ticks, 
	yticklabel pos=right,
	ymajorgrids,
	every major grid/.style={opacity=0.7},
	xtick={1,2,3,4,5,6,7,8,9,10},
	ytick={1e1,1e0,1e-1,1e-2,1e-3,1e-4,1e-5,1e-6,1e-7,1e-8,1e-9,1e-10,1e-11},
	legend style={
		cells={anchor=east},
		legend pos=south west,
	}
]

\addplot[mark=*,mark options={scale=1,solid}, line width=1.0pt, dash pattern=on 4pt off 2pt on 1pt off 2pt on 1pt off 2pt] 
table [x=a, y=b, col sep=comma] {data1.csv};
\addplot[mark=o,mark options={solid}, line width=1.0pt, dash pattern=on 4pt off 2pt on 1pt off 2pt] 
table [x=a, y=c, col sep=comma] {data1.csv};
\addplot[mark=square*, mark options={solid}, line width=1.0pt, dash pattern=on 4pt off 1pt on 4pt off 4pt] 
table [x=a, y=d, col sep=comma] {data1.csv};
\addplot[mark=square, mark options={solid}, line width=1.0pt, dash pattern=on 4pt off 2pt] 
table [x=a, y=e, col sep=comma] {data1.csv};
\addplot[mark=triangle*, mark options={solid}, line width=1.0pt] 
table [x=a, y=f, col sep=comma] {data1.csv};

\legend{
	{$\sourcetime=4\,\stdgauss$},
	{$\sourcetime=5\,\stdgauss$},
	{$\sourcetime=6\,\stdgauss$},
	{$\sourcetime=7\,\stdgauss$},
	{$\sourcetime=8\,\stdgauss$},
} 

\end{semilogyaxis}

\end{tikzpicture}

%% file: figures/Genergy.tex
\begin{tikzpicture}[scale=1.0]
	\def\xm{6.3}
	\def\tO{5.0}
	\def\tA{5.5}
	\def\tB{6.5}
	\def\tC{7.5}
	\def\tD{8.5}
	\def\tE{9.5}
	\def\tF{10.5}
	\def\omegaA{1.0}
	\def\omegaB{2.0}
	\def\omegaC{3.0}

	\pgfplotsset{every tick label/.append style={font=\small}} 
	\begin{axis}[
		width=15cm,
		height=6cm,
		xmin=0, xmax=\xm,
		ymin=0, ymax=1.0,
		domain=-\xm:\xm,
		grid=major,
	    every major grid/.style={opacity=0.7}, 
		xtick={0.0,1,2,3,4,5,6,7},
		ytick={0.47,0.94},
		samples=200, 
		legend cell align=left,
		legend style={at={(0.78,0.50)},anchor=west},
		axis on top=true,
		axis x line=middle,
		axis y line=middle,
		xlabel=$t(\textrm{s})$,
		ylabel=\large (a),
		title=$G^{r}$ energy,
		every axis y label/.style={at={(ticklabel* cs:1.01)}, anchor=south}, 
		every axis x label/.style={at={(ticklabel* cs:1.01)}, anchor=west}
	]
	\addplot[line width=1.8pt, dash pattern=on 4pt off 2pt on 1pt off 2pt on 1pt off 2pt]
(x,{exp(-(x^2)/2)*sqrt(pi/2)*(3+3*exp((x^2)/2)-6*x^2+x^4)/(8*\omegaA)});
	\addplot[line width=1.8pt, dash pattern=on 4pt off 2pt] 
(x,{exp(-(x^2)/2)*sqrt(pi/2)*(-3+3*exp((x^2)/2)+6*x^2-x^4)/(8*\omegaA)});
	\addplot[line width=1.8pt]
(x,{3*sqrt(pi/2)/(4*\omegaA)});
	\addplot[grays1, line width=1.0pt, dash pattern=on 4pt off 2pt on 1pt off 2pt on 1pt off 2pt]
(x,{exp(-(\omegaB^2*x^2)/2)*sqrt(pi/2)*(3+3*exp((\omegaB^2*x^2)/2)-6*\omegaB^2*x^2+\omegaB^4*x^4)/(8*\omegaB)});
	\addplot[grays1,line width=1.0pt, dash pattern=on 4pt off 2pt]
(x,{exp(-(\omegaB^2*x^2)/2)*sqrt(pi/2)*(-3+3*exp((\omegaB^2*x^2)/2)+6*\omegaB^2*x^2-\omegaB^4*x^4)/(8*\omegaB)});
	\addplot[grays1,line width=1.0pt]
(x,{3*sqrt(pi/2)/(4*\omegaB)});

    \legend{
		{$K_{c,\peakfreq=1}^{G,r}(t)$},
		{$U_{c,\peakfreq=1}^{G,r}(t)$},
		{$E_{c,\peakfreq=1}^{G,r}(t)$},
		{$K_{c,\peakfreq=2}^{G,r}(t)$},
		{$U_{c,\peakfreq=2}^{G,r}(t)$},
		{$E_{c,\peakfreq=2}^{G,r}(t)$},
	} 
	\end{axis}
\end{tikzpicture}

%% file: figures/Senergy.tex
\begin{tikzpicture}[scale=1.0,     
		declare function={erf(\x)=%
      	(1+(e^(-(\x*\x))*(-265.057+abs(\x)*(-135.065+abs(\x)%
      	*(-59.646+(-6.84727-0.777889*abs(\x))*abs(\x)))))%
      	/(3.05259+abs(\x))^5)*(\x>0?1:-1);}]
	\def\xm{12.0}
	\def\tO{5.0}
	\def\tA{5.5}
	\def\tB{6.5}
	\def\tC{7.5}
	\def\tD{8.5}
	\def\tE{9.5}
	\def\tF{10.5}
	\def\omegaA{1.0}
	\def\omegaB{2.0}
	\def\omegaC{3.0}

	\pgfplotsset{every tick label/.append style={font=\small}} 
	\begin{axis}[
		width=15cm,
		height=6cm,
		xmin=0, xmax=\xm,
		ymin=0, ymax=1.0,
		domain=0:\xm,
		grid=major,
	    every major grid/.style={opacity=0.7},
		xtick={0.0,1,3,5,7,9,11},
		ytick={0.47,0.94},
		samples=200, 
		legend cell align=left,
		legend style={at={(0.03,0.65)},anchor=west},
		axis on top=true,
		axis x line=middle,
		axis y line=middle,
		xlabel=$t(\textrm{s})$,
		ylabel=\large (b),
		title=$S^{r}$ energy,
		every axis y label/.style={at={(ticklabel* cs:1.01)}, anchor=south}, 
		every axis x label/.style={at={(ticklabel* cs:1.01)}, anchor=west}
	]
    \addplot[line width=1.8pt, dash pattern=on 4pt off 2pt]
(x,{(1/(16*\omegaA)) * (  3*sqrt(pi/2)*(erf((x-\tO)*\omegaA/sqrt(2))+1) - (x-\tO)*\omegaA*exp(-((x-\tO)^2*\omegaA^2)/2)*((x-\tO)^2*\omegaA^2 -1) )});
	\addplot[line width=1.8pt]
(x,{(1/(8*\omegaA)) * (  3*sqrt(pi/2)*(erf((x-\tO)*\omegaA/sqrt(2))+1) - (x-\tO)*\omegaA*exp(-((x-\tO)^2*\omegaA^2)/2)*((x-\tO)^2*\omegaA^2 -1) )});
	\addplot[grays1, line width=1.0pt, dash pattern=on 4pt off 2pt]
(x,{(1/(16*\omegaB)) * (  3*sqrt(pi/2)*(erf((x-\tO)*\omegaB/sqrt(2))+1) - (x-\tO)*\omegaB*exp(-((x-\tO)^2*\omegaB^2)/2)*((x-\tO)^2*\omegaB^2 -1) )});
	\addplot[grays1, line width=1.0pt]
(x,{(1/(8*\omegaB)) * (  3*sqrt(pi/2)*(erf((x-\tO)*\omegaB/sqrt(2))+1) - (x-\tO)*\omegaB*exp(-((x-\tO)^2*\omegaB^2)/2)*((x-\tO)^2*\omegaB^2 -1) )});

    \legend{
		{$K_{c,\peakfreq=1}^{S,r}(t)=U_{c,\peakfreq=1}^{S,r}(t)$},
		{$E_{c,\peakfreq=1}^{S,r}(t)$},
		{$K_{c,\peakfreq=2}^{S,r}(t)=U_{c,\peakfreq=2}^{S,r}(t)$},
		{$E_{c,\peakfreq=2}^{S,r}(t)$},
	} 
	\end{axis}
\end{tikzpicture}

%% file: figures/GSnumData.tex
\begin{minipage}[c][1\totalheight][t]{0.45\columnwidth}%
\begin{center}(a)

\begin{tikzpicture}[scale=0.9]
	\begin{axis}[
			view={15}{50},
			zmin=-0.3, zmax=0.65,
			xtick={-10,-5,0,5,10},
			ytick={0,5,10,15,20},
		    xlabel=$x$(m), 
		    ylabel=$t$(s),
			title=\Large $u^{(g)}$
		]
		\addplot3[mark=o, mark size=0.6] file {data2.csv};
	\end{axis}
\end{tikzpicture}\end{center}%
\end{minipage}%
\begin{minipage}[c][1\totalheight][t]{0.45\columnwidth}%
\begin{center}(b)

\begin{tikzpicture}[scale=0.9]
	\begin{axis}[
			view={15}{50},
			zmin=-0.3, zmax=0.65,
			xtick={-10,-5,0,5,10},
			ytick={0,5,10,15,20},
		    xlabel=$x$(m), 
		    ylabel=$t$(s),
			title=$\Large u^{(s)}$
		]
		\addplot3[mark=o, mark size=0.6] file {data3.csv};
	\end{axis}
\end{tikzpicture}\end{center}%
\end{minipage}

\begin{minipage}[c][1\totalheight][t]{0.9\columnwidth}%
\begin{center}(c)

\begin{tikzpicture}[scale=1.0]
\begin{semilogyaxis}[
	width=13cm,
	height=4.5cm,
	xmin=0.0, xmax=12.0,
	xlabel=$t(\textrm{s})$,
	title={$D^{G^{r},S^{r}}_{1,1,\sourcetime}\left(t\right)$},
	ymajorgrids,
	every major grid/.style={opacity=0.7},
	legend style={
		cells={anchor=east},
		legend pos=south west,
	}
]
\addplot[mark=o,mark options={scale=1,solid}, line width=1.0pt] 
table [x=a, y=b, col sep=comma] {data4.csv};
\end{semilogyaxis}
\end{tikzpicture}\end{center}%
\end{minipage}

%% file: main.bbl
\begin{thebibliography}{25}
\providecommand{\natexlab}[1]{#1}
\providecommand{\url}[1]{\texttt{#1}}
\expandafter\ifx\csname urlstyle\endcsname\relax
  \providecommand{\doi}[1]{doi: #1}\else
  \providecommand{\doi}{doi: \begingroup \urlstyle{rm}\Url}\fi

\bibitem[Habets(2006)]{habets2006room}
Emanuel~AP Habets.
\newblock Room impulse response generator.
\newblock \emph{Technische Universiteit Eindhoven, Tech. Rep}, 2\penalty0
  (2.4):\penalty0 1, 2006.

\bibitem[Allen and Berkley(1979)]{allen1979image}
Jont~B Allen and David~A Berkley.
\newblock Image method for efficiently simulating small-room acoustics.
\newblock \emph{The Journal of the Acoustical Society of America}, 65\penalty0
  (4):\penalty0 943--950, 1979.

\bibitem[Ward and Abhayapala(2001)]{ward2001reproduction}
Darren~B Ward and Thushara~D Abhayapala.
\newblock Reproduction of a plane-wave sound field using an array of
  loudspeakers.
\newblock \emph{IEEE Transactions on speech and audio processing}, 9\penalty0
  (6):\penalty0 697--707, 2001.

\bibitem[Tiwana et~al.(2017)Tiwana, Ahmed, Mann, and Naqvi]{tiwana2017point}
MH~Tiwana, Shakeel Ahmed, AB~Mann, and QA~Naqvi.
\newblock Point source diffraction from a semi-infinite perfect electromagnetic
  conductor half plane.
\newblock \emph{Optik-International Journal for Light and Electron Optics},
  135:\penalty0 1--7, 2017.

\bibitem[Vlaar(1966)]{vlaar1966field}
NJ~Vlaar.
\newblock The field from an sh point source in a continuously layered
  inhomogeneous half-space ii. the field in a half-space.
\newblock \emph{Bulletin of the Seismological Society of America}, 56\penalty0
  (6):\penalty0 1305--1315, 1966.

\bibitem[Vidale and Helmberger(1988)]{vidale1988elastic}
John~E Vidale and Donald~V Helmberger.
\newblock Elastic finite-difference modeling of the 1971 san fernando,
  california earthquake.
\newblock \emph{Bulletin of the Seismological Society of America}, 78\penalty0
  (1):\penalty0 122--141, 1988.

\bibitem[Frankel(1993)]{frankel1993three}
Arthur Frankel.
\newblock Three-dimensional simulations of ground motions in the san bernardino
  valley, california, for hypothetical earthquakes on the san andreas fault.
\newblock \emph{Bulletin of the Seismological Society of America}, 83\penalty0
  (4):\penalty0 1020--1041, 1993.

\bibitem[Evans(1997)]{evans1997handbook}
Brian~J Evans.
\newblock \emph{A handbook for seismic data acquisition in exploration}.
\newblock Society of exploration geophysicists, 1997.

\bibitem[Meunier(2011)]{meunier2011seismic}
Julien Meunier.
\newblock \emph{Seismic acquisition from yesterday to tomorrow}.
\newblock Society of Exploration Geophysicists, 2011.

\bibitem[Baysal et~al.(1983)Baysal, Kosloff, and Sherwood]{baysal1983reverse}
Edip Baysal, Dan~D Kosloff, and John~WC Sherwood.
\newblock Reverse time migration.
\newblock \emph{Geophysics}, 48\penalty0 (11):\penalty0 1514--1524, 1983.

\bibitem[McMechan(1989)]{mcmechan1989review}
George~A McMechan.
\newblock A review of seismic acoustic imaging by reverse-time migration.
\newblock \emph{International Journal of Imaging Systems and Technology},
  1\penalty0 (1):\penalty0 18--21, 1989.

\bibitem[Virieux and Operto(2009)]{virieux2009overview}
Jean Virieux and St{\'e}phane Operto.
\newblock An overview of full-waveform inversion in exploration geophysics.
\newblock \emph{Geophysics}, 74\penalty0 (6):\penalty0 WCC1--WCC26, 2009.

\bibitem[Cohen and Bleistein(1979)]{cohen1979velocity}
Jack~K Cohen and Norman Bleistein.
\newblock Velocity inversion procedure for acoustic waves.
\newblock \emph{Geophysics}, 44\penalty0 (6):\penalty0 1077--1087, 1979.

\bibitem[Alford et~al.(1974)Alford, Kelly, and Boore]{alford1974accuracy}
RM~Alford, KR~Kelly, and D~Mt Boore.
\newblock Accuracy of finite-difference modeling of the acoustic wave equation.
\newblock \emph{Geophysics}, 39\penalty0 (6):\penalty0 834--842, 1974.

\bibitem[D'Alembert(1747)]{drecherches}
Jean le~Rond D'Alembert.
\newblock Recherches sur la courbe que forme une corde tendu{\"e} mise en
  vibration: Suite.
\newblock \emph{Histoire de l'Acad{\'e}mie Royale des Sciences et des Belles
  Lettres de Berlin}, page 220, 1747.

\bibitem[Miersemann(2012)]{miersemann2012partial}
Erich Miersemann.
\newblock \emph{Partial Differential Equations Lecture Notes}.
\newblock Citeseer, 2012.

\bibitem[Dr{\'a}bek and Holubov{\'a}(2014)]{drabek2014elements}
Pavel Dr{\'a}bek and Gabriela Holubov{\'a}.
\newblock \emph{Elements of partial differential equations}.
\newblock Walter de Gruyter GmbH \& Co KG, 2014.
\newblock pp. 73.

\bibitem[Mallat(1999)]{mallat1999wavelet}
St{\'e}phane Mallat.
\newblock \emph{A wavelet tour of signal processing}.
\newblock Academic press, 1999.

\bibitem[Gholamy and Kreinovich(2014)]{gholamy2014ricker}
Afshin Gholamy and Vladik Kreinovich.
\newblock Why ricker wavelets are successful in processing seismic data:
  Towards a theoretical explanation.
\newblock In \emph{2014 IEEE Symposium on Computational Intelligence for
  Engineering Solutions (CIES)}, pages 11--16. IEEE, 2014.

\bibitem[Wang(2015{\natexlab{a}})]{wang2015generalized}
Yanghua Wang.
\newblock Generalized seismic wavelets.
\newblock \emph{Geophysical Journal International}, 203\penalty0 (2):\penalty0
  1172--1178, 2015{\natexlab{a}}.

\bibitem[Wang(2015{\natexlab{b}})]{wang2015frequencies}
Yanghua Wang.
\newblock Frequencies of the ricker wavelet.
\newblock \emph{Geophysics}, 80\penalty0 (2):\penalty0 A31--A37,
  2015{\natexlab{b}}.

\bibitem[Wang(2015{\natexlab{c}})]{wang2015ricker}
Yanghua Wang.
\newblock The ricker wavelet and the lambert w function.
\newblock \emph{Geophysical Journal International}, 200\penalty0 (1):\penalty0
  111--115, 2015{\natexlab{c}}.

\bibitem[Duffin(1970)]{duffin1970equipartition}
Richard~James Duffin.
\newblock Equipartition of energy in wave motion.
\newblock \emph{Journal of Mathematical Analysis and Applications}, 32\penalty0
  (2):\penalty0 386--391, 1970.

\bibitem[Kukudzhanov(2013)]{kukudzhanov2013numerical}
Vladimir~N Kukudzhanov.
\newblock \emph{Numerical continuum mechanics}, volume~15.
\newblock Walter de Gruyter, 2013.

\bibitem[Courant et~al.(1967)Courant, Friedrichs, and Lewy]{courant1967partial}
Richard Courant, Kurt Friedrichs, and Hans Lewy.
\newblock On the partial difference equations of mathematical physics.
\newblock \emph{IBM journal}, 11\penalty0 (2):\penalty0 215--234, 1967.

\end{thebibliography}
